\documentclass[conference,letterpaper]{IEEEtran}

\usepackage[utf8]{inputenc} 
\usepackage[T1]{fontenc}
\usepackage{url}
\usepackage{ifthen}
\usepackage{cite}
\usepackage[cmex10]{amsmath} 
\usepackage{xspace}
\usepackage{booktabs}
\usepackage{graphicx}
\usepackage{amsthm}
\usepackage{algorithm}
\usepackage{algorithmic}
\usepackage{booktabs}
\usepackage{amsmath,amssymb}
\usepackage{mathtools}
\usepackage{xcolor}
\usepackage{listings}
\usepackage{hyperref}
\usepackage{tikz}
\usepackage{tikz-cd}
\usepackage{float}
\usepackage{color}

\newtheorem{theorem}{Theorem}
\newtheorem{lemma}[theorem]{Lemma}
\newtheorem{definition}{Definition}[section]

\newcommand{\bra}[1]{\left \langle #1 \right \rvert}
\newcommand{\ket}[1]{\left \rvert #1 \right \rangle}

\newcommand\norm[1]{\lVert#1\rVert}
\newcommand{\remove}[1]{}

\begin{document}
\title{Secure Composition of Quantum Key Distribution and Symmetric Key Encryption} 

\author{%
 \IEEEauthorblockN{Kunal Dey, Reihaneh Safavi-Naini}
 \IEEEauthorblockA{
                   University of Calgary,
                   AB, Canada\\
                   }%
}


\maketitle


\begin{abstract}
 
Quantum key distribution (QKD) allows Alice and Bob to share a secret key over an insecure channel with proven information-theoretic security against an adversary whose strategy is bounded only by the laws of physics. Composability-based security proofs of QKD ensure that using the established key with a one-time-pad encryption scheme provides {\em information theoretic secrecy} for the message. In this paper, we consider the problem of using the QKD established key with a secure symmetric key-based encryption algorithm and use an approach based on {\em hybrid encryption} to provide a proof of security for the composition. 

Hybrid encryption was first proposed as a public key cryptographic algorithm with proven security for messages of unrestricted length. We use an extension of this framework to correlated randomness setting (Sharifian et al. in ISIT 2021) to propose a {\em quantum-enabled Key Encapsulation Mechanism (qKEM)} and {\em quantum-enabled hybrid encryption (qHE)}, and prove a composition theorem for the security of the qHE. We construct a qKEM with proven security using an existing  QKD (Portmann et al. in Rev. of Mod. Physics 2022). Using this qKEM with a secure {\em Data Encapsulation Mechanism (DEM)}, that can be constructed using a one-time symmetric key encryption scheme, results in an efficient encryption system for unrestricted length messages with {\em proved} security against an adversary with access to  efficient computations on a quantum computer (i.e. post-quantum secure encryption without using  any computational assumptions.)
\end{abstract}

\section{Introduction}

One of the fundamental problems in cryptography is {\em secret key establishment (SKA)} where two parties, Alice and Bob, interact over an insecure channel to establish a shared secret key that is completely unknown to the adversary, Eve. Key establishment is only possible if one makes additional assumptions. 
The two main types of assumptions in classical cryptography are: (i) assuming that the adversary is computationally bounded and basing security on ``hard'' mathematical problems for which no efficient solution is known \cite{1055638}, or (ii) considering a computationally unbounded adversary (information-theoretic security) but assuming Alice and Bob have private samples of correlated random variables that are partially known to Eve \cite{maurer1993secret}. A third approach that also guarantees information-theoretic security for the established key is {\em Quantum Key Distribution (QKD)} which relies on quantum theoretic assumptions only. The last two approaches are particularly attractive because security is unconditional and does not depend on advances in computing such as the development of quantum computers.

Security proof of QKD has attracted much research, starting with early security proofs \cite{shor2000simple, mayers2001unconditional}, followed by revising security conditions to comprehensively capture the adversary’s information about the key \cite{lo2005efficient, renner2004smooth, tomamichel2012tight, renner2008security}, 
and using simulation-based security \cite{ben2005universal, unruh2004simulatable, muller2009composability} to allow composition with other cryptograpic protocols.

In practice an established key is used to perform a symmetric key cryptographic operation, for example, encrypting messages of unrestricted length \footnote{We follow the terminology used in \cite{cramer2003design} which is recalled in subsection \ref{HE PK}.} that Alice wants to send to Bob, or vice versa. A natural question is: what can be said about the security of the encrypted message once the established key is used in a secure encryption algorithm?

Simaluation-based proofs allow the established key to be used in an application that uses the same computational model of QKD, an adversary with unbounded computation. This leaves the question of using the QKD established key in a computationally secure systems, for example a secure encryption.

Cramer and Shoup \cite{cramer2003design} considered this problem when the key establishment is a public key algorithm and security is against a computationally bounded adversary. They introduced the hybrid public-key encryption (HPKE) scheme with two building blocks: a Key Encapsulation Mechanism (KEM), which is a public key primitive for shared key establishment, and a Data Encapsulation Mechanism (DEM), which is a symmetric key encryption system. They proved a composition theorem that relates the security of HPKE to the security of its two building blocks.

Sharifian et al. \cite{sharifian2021information} extended the KEM/DEM framework to the {\em correlated randomness model} and introduced {\em information-theoretic KEM (iKEM)} that can be used with a DEM with computational security to obtain a {\em computationally secure hybrid encryption (HE) scheme} for unrestricted messages {\em without using any computational assumptions.} They showed that iKEM can be constructed using an information-theoretic one-way secure key agreement. This extension is particularly attractive because the security of the encryption scheme effectively relies on the computational security of the DEM only (iKEM has information-theoretic security). A secure DEM, in both HPKE and HE, is a one-time secure symmetric key encryption scheme that can be constructed using AES (for example in counter mode). The best-known quantum attack against such a DEM is key search using Grover’s algorithm \cite{grover1996fast}.

{\em Our work:} 
Our goal in this paper is to extend the framework of Sharifian et al. to KEMs that use quantum communication for correlation generation and key establishment, with the ultimate goal of proving the security of the composition of a QKD and a secure symmetric key encryption system.

{\em From iKEM to Quantum-enabled KEM (qKEM).} The HE framework of \cite{sharifian2021information} and the iKEM setting, is a promising approach for composing a QKD-generated key and a symmetric key encryption system. In the iKEM setting, there is a public  distribution $P(XYZ)$ that, during initialization, is sampled by a trusted sampling process to generate private samples $(x,y,z)$ for Alice, Bob and Eve, respectively. For instance, in Maurer’s satellite scenario \cite{maurer1993secret}, $P(XYZ)$ is generated by a satellite broadcasting a random beacon that is received by Alice, Bob, and Eve through their respective noisy channels. Inspired by the BB84 and prepare-and-measure QKD protocols, one can generate $P(XYZ)$ by using a quantum channel and a classical public channel, as follows. Alice will use the quantum communication channel to send a sequence of quantum states, that is constructed by Wiesner coding of a random binary string. The quantum channel is controlled by Eve who has unbounded quantum resources. Bob receives a sequence of modified quantum states and, using its own measurements, obtains a random string that is {\em correlated} with Alice’s string. Note, however, that this {\em correlation generation is in the presence of Eve}, while iKEM setup assumes $P(XYZ)$ is publicly known and the sampling process is trusted. 

After communication over the quantum channel, Alice and Bob use {\em sifting} and {\em error estimation} steps over a classical channel  (see Section \ref{QKD}), to either abort or obtain a binary string each and an estimate of the number of ``mismatches'' between the two strings. This estimate is an indication of the level of the adversary’s interaction with the channel. The leakage of information to the adversary depends on Eve’s quantum strategy. We consider an adversary that uses {\em collective strategy} (see section \ref{QKD}), allowing it to optimize its measurement by taking into account communications over classical channels.

A significant advantage of this method of correlation generation compared to the original iKEM setting is that in the iKEM case, there is an implicit assumption that Eve’s information about Alice and Bob’s random variables is known. (In Maurer’s satellite scenario this means Eve’s channel and its receiver (including antenna) are known.) Using quantum communication as described above, however, allows Alice and Bob to correctly estimate Eve’s information about the key.

In section \ref{qKEM} we define {\em qKEM (quantum enabled KEM)}, with  information-theoretic security and, with an untrusted correlation generation. We define a key indistinguishability security game (Figure \ref{qKEM game})  against an  adversary that has unbounded quantum resources, can interact with the quantum channel and eavesdrops on communication over the public channel. In Theorem \ref{Thm:qHE compos} we prove security of the composition of a  secure qKEM  and a  secure DEM (one-time secure symmetric encryption -- see security game in Figure \ref{qHE-OT:game}), resulting in a secure {\em quantum-enabled hybrid  encryption (qHE)} with security against {\em an adversary with bounded quantum computation resources} (including efficient computation using a quantum computer)

{\em Constructions:} In Section \ref{instantiation:qKEM} we construct a qKEM using the prepare-and-measure QKD in \cite{portmann2022security}.
Theorem \ref{secproof:cqaqkem} proves security of the construction, that is identical to the QKD, as a qKEM with security defined by the game in Figure \ref{qKEM game}. 

{\em What we achieve:} qHE uses a qKEM and a secure DEM to encrypt messages of unrestricted length with proven security against an adversary with access to a quantum computer and without using any computational assumptions. The encryption is efficient as message encryption uses secure symmetric key-based encryption.

Alternative approaches to achieving similar security and efficiency for encrypted messages are (i) using a post-quantum secure KEM such as CRYSTALS-Kyber that uses computational assumption (and trusted public key) \cite{bos2018crystals}, with a similar secure DEM, or, (ii) use an iKEM with a trusted correlation and sampling process and a similar DEM. The drawback of (i) is reliance on unproven computational assumptions, and the drawback of (ii) is unverifiable assumptions about Eve's information. Our work can be seen as strengthening the latter approach by removing assumptions about Eve's information.

{\em Paper organization}. Section \ref{prelim} is preliminaries. Section \ref{Hybrid Enc} introduces the hybrid encryption framework and its security, and its extension to correlated randomness setting. Section \ref{QKD} is on QKD. Section \ref{qKEM} and \ref{qHE} introduce qKEM and qHE respectively and prove the composition theorem of qHE. Section \ref{instantiation:qKEM} is the construction of a qKEM and a qHE, and Section \ref{conclusion} concludes the paper.

\section{Preliminaries}
\label{prelim}

\textbf{Notations}: 
We represent random variables (RVs) using uppercase letters (e.g., X), and their specific values or realizations with lowercase letters (e.g., $x$). Sets are denoted with calligraphic letters, such as $\mathcal{X}$, and the size of a set $\mathcal{X}$ is expressed as 
$|\mathcal{X}|$. Vectors are denoted in boldface, and  $\mathbf{X}$ = $X^n$ = $(X_1,..., X_n)$ is a vector consisting of n random variables, with its specific realization denoted by $\mathbf{x} = (x_1,...,x_n)$. $U_{\ell}$ denotes a uniformly distributed random variable over the set $\{0,1\}^{\ell}$. $v \longleftarrow A^O(x,y,...)$ denotes that an algorithm A has access to the oracle (classical) O, and with inputs $x,y,...$ it outputs $v$. we use, ``$\xlongleftarrow{\$}$" to assign to a variable either a uniformly sampled value from
a set or the output of a randomized algorithm.\\
\noindent\textbf{Statistical distance and negligible function}: For distributions $P_X$ and $P_Y$ over the same domain $\mathcal{X}$ , we denote statistical distance (SD) between two RVs $X$ and $Y$ as,
\begin{align*}
    SD(X,Y) = \frac{1}{2} \sum_{x \in \mathcal{X}} |P_X(x) - P_Y(x)|
\end{align*}
Where $P_X(x) = \Pr(X=x)$. Another useful representation of SD is,
\begin{align*}
  SD(X,Y) = \max_{\mathcal{W} \subset \mathcal{X}} (\Pr(X \in \mathcal{W}) -  \Pr(Y \in \mathcal{W})) 
\end{align*}
where $\Pr(X \in \mathcal{W}) = \sum_{w \in \mathcal{W}} \Pr(X = w)$. The class of all negligible functions denoted as {\em NEGL} is the collection of all functions $f : \mathbb{N} \longrightarrow \mathbb{R}_{\ge 0}$ such that for all positive polynomial functions $P(\cdot)$, there exist $\beta >0$ such that,
$f(\alpha) \le \frac{1}{P(\alpha)}$ for all $\alpha > \beta$.

\begin{definition}[Two-universal hash function]
\label{2-UHF}
Let $\mathcal{F}$ be a family of functions from  
a set $\mathcal{U}$ to $\mathcal{V}$. $p_F$ is a probability distribution on $\mathcal{F}$. The two-universal hash function (2-UHF) is defined by the pair $(\mathcal{F}, p_F)$ such that,
\[\Pr_{f\in {\mathcal{F}}}[f(x) = f(x^{\prime})] \leq \frac{1}{|\mathcal{\mathcal{V}}|}\]
for any $x, x^{\prime} \in \mathcal{U}$ with $x \neq x^{\prime}$.
\end{definition}
Two-universal hash functions \cite{tomamichel2011leftover} can serve as a quantum-proof strong randomness extractor \footnote{Another example is Trevisan's extractor, as discussed in \cite{de2012trevisan}, for instance.}.

\subsection{Quantum Background}
Pure state of a quantum system is represented by a vector in a Hilbert space $\mathcal{H}$. A vector in $\mathcal{H}$ is denoted by $\ket{\psi}$ (called “ket” $\psi$). In a two-dimensional Hilbert space, The simplest example of a quantum system is a qubit $\ket{b}$ whose state is  expressed as a superposition of two basis vectors: $\ket{0} = \begin{pmatrix}
    1 \\
    0
\end{pmatrix}$ and $\ket{1} = \begin{pmatrix}
    0 \\
    1
\end{pmatrix}$ i.e, $\ket{b} = \alpha \ket{0} + \beta \ket{1}$, $\alpha,\beta \in \mathbb{C}$,  where $|\alpha|^2 + |\beta|^2 = 1$, and $\mathbb{C}$ is the set of complex numbers. A quantum system is represented by an ensemble of pure states $\{p_i, \ket{\psi_i}\}^k_{i=1}$, where the states $\ket{\psi_1},\cdots, \ket{\psi_k}$ are associated with probabilities $p_1,\cdots,p_k$. The system can be described by a density operator (or density matrix) $\rho$, defined as:
\begin{align*}
    \rho = \sum_{i=1}^k p_i \ket{\psi_i}\bra{\psi_i}
\end{align*}
\noindent $\bra{\psi_i}$ is the dual vector of $\ket{\psi_i}$ in the dual 
space $\mathcal{H}^{*}$. A density operator is a normalized operator (with trace equal to 1) that is Hermitian and positive semi-definite. 

For a quantum state $\rho$ over some Hilbert space $\mathcal{H}$, $\norm{\rho}_1 = \text{Tr} |\rho|$ is the {\em trace norm of $\rho$}, where $|\rho|$ is the positive square root of $\rho^{\dagger}\rho$. The {\em trace distance} (TD) between two states $\rho$ and $\rho^{\prime}$ is represented by $\frac{1}{2} \norm{\rho - \rho^{\prime}}_1$.

\noindent \textbf{Quantum Channels and Measurement}: For a Hilbert space $\mathcal{H}$, we define $\mathcal{B}(\mathcal{H})$ as the space of bounded linear operators on $\mathcal{H}$. A quantum channel between two quantum systems $A$ and $B$ is a completely positive trace-preserving (CPTP) from $\mathcal{B}(\mathcal{H}_A)$ to $\mathcal{B}(\mathcal{H}_B)$. A {\em generalized measurement} acting on the state space of the system A being measured is characterized by a collection of linear operators $\{M^A_x\}$, where $x \in \mathcal{X}$ are possible classical outcome after implementing measurement on some quantum state $\rho$. This collection must adhere to the following completeness relation:
\begin{align*}
    \sum_x {(M^A_x)}^{\dagger}{M^A_x} = \mathbb{I}
\end{align*}
where ${(M^A_x)}^{\dagger}$ is the adjoint operator of $M^A_x$ and $\mathbb{I}$ is the identity matrix.  The probability of a specific outcome $x$ is determined by $\Pr_{\rho}(x) = \text{Tr}({(M^A_x)}^{\dagger} M^A_x \rho)$, where $\text{Tr}(\cdot)$ is the trace function. A {\em positive operator-valued measure} (POVM) on A is characterized by a collection of positive operators $\{E^A_x\}$ where $E^A_x$ is computed as ${(M^A_x)}^{\dagger} M^A_x$ with the completeness relation $\sum_x E^A_x = \mathbb{I}$. The probability of a specific outcome  $x$ is formally expressed as $\Pr_{\rho}(x) = \text{Tr}({E^A_x} \rho)$.

\noindent \textbf{Entropy}: In the asymptotic scenario, where a process is repeated many times independently, the primary relevant entropy measure is the {\em Shannon entropy} (or its quantum equivalent, the {\em Von Neumann entropy}). However, this changes in non-asymptotic scenarios or when the independence assumption is removed. In such cases, Shannon and Von Neumann entropies no longer accurately characterize operational quantities, requiring the use of more general entropy measures. Smooth min- and max-entropies are particularly versatile as they do not rely on the repetition of random processes. They apply to situations where a source emits only a single piece of information or a channel is used just once. In this work, we adopt the conditional quantum max- and min-entropy frameworks, as introduced in \cite{renner2008security, konig2009operational}.

\begin{definition}[Quantum conditional min-entropy]
The quantum conditional min-entropy of a bipartite state $\rho = \rho_{AB}$ of two quantum systems A and B is defined as:
\begin{equation*}
    H_\text{min} (A|B)_\rho = - \inf_{ \sigma_B} D_{\infty} (\rho_{AB} || \sigma_B)
\end{equation*}
where the infimum is over all density operators $\sigma_B$ on B  and $D_{\infty} (\rho_{AB} || \sigma_B) = D_{\infty} (\rho_{AB} || I_A \otimes \sigma_B) = \inf \{\lambda : \rho_{AB} \leq 2^{\lambda} (I_A \otimes \sigma_B)\}$.

\end{definition}

Let $S_{\le}(\mathcal{H})$ be the collection of all states over the Hilbert space $\mathcal{H}$ such that the trace of that state is less or equal to 1. The $\varepsilon$-ball of a state $\rho \in S_{\le}(\mathcal{H})$ is denoted by,
\[\mathcal{B}^{\varepsilon}(\rho) = \{\rho' \in S_{\le}(\mathcal{H}) : P (\rho, \rho') \leq \varepsilon\}\]
where $ P (\rho, \rho^{\prime})$ is the purified distance \cite{nielsen2010quantum} between $\rho$ and $\rho^{\prime}$. The smooth conditional min/max-entropy of a state $\rho$ is defined as the min/max-entropy of an “optimal” state $\rho' $  within $\mathcal{B}^{\varepsilon}(\rho)$, where $\varepsilon$ is referred to as the smoothness parameter.
\begin{definition}[Smooth conditional min-entropy and max-entropy]
\label{def:smooth-min}
For a joint quantum state $\rho_{AB}$ the smooth conditional min-entropy is defined by,
\begin{equation*}
    H_\text{min}^{\varepsilon}(A|B)_\rho = \max_{\rho' \in {\cal B}^{\varepsilon}(\rho) } H_\text{min} (A|B)_{\rho'}
\end{equation*}
\end{definition}
The smooth conditional max-entropy is the dual \cite{konig2009operational, tomamichel2010duality} of $H_\text{min}^{\varepsilon}(A|B)$ with regards to any purification $\rho_{ABC}$ of $\rho_{AB}$ in the sense that,
\begin{align*}
    H_\text{max}^{\varepsilon}(A|B)_\rho = - H_\text{min}^{\varepsilon}(A|C)_\rho
\end{align*}

\noindent \textbf{Quantum Algorithms} \cite{ananth2022cryptography}: A quantum algorithm is a family of generalized quantum circuits $\{\textit{C}^{\text{Q}}_\lambda\}_{\lambda \in \mathbb{N}}$ constructed using a universal quantum gate set. The size of a quantum circuit is defined as the total number of gates and the number of input/output qubits. A quantum algorithm $\textit{C} = \{\textit{C}^{\text{Q}}_\lambda\}_{\lambda \in \mathbb{N}}$ is a {\em quantum polynomial-time (QPT) algorithm} if the size of each circuit is bounded by a polynomial $p(\lambda)$ for some polynomial $p(x)$.

\section{Hybrid Encryption using KEM/DEM }
\label{Hybrid Enc}

A {\em hybrid encryption} (HPKE) scheme is a public-key encryption (PKE) scheme consisting of two components: (i) a Key Encapsulation Mechanism (KEM), that is used to encrypt a symmetric key that can be decrypted by Bob, establishing a shared key between Alice and Bob; and (ii) a Data Encapsulation Mechanism (DEM), a special symmetric key encryption that can be used to encrypt arbitrary long messages. In this section, we discuss two variants of the hybrid encryption system: the public-key variant of HPKE \cite{cramer2003design} and HE in the preprocessing model \cite{sharifian2021information}. Throughout, we use $\lambda$ to denote the system's security parameter.

\subsection{Hybrid Encryption in Public Key Model}
\label{HE PK}

Public key KEM was formalized in \cite{cramer2003design} where its composition with a DEM was proved to result in a public key HPKE.

\begin{definition}[Key Encapsulation Mechanism (KEM)]\label{def:kem} 
A key encapsulation mechanism {\sf KEM = (KEM.Gen, KEM.Enc, KEM.Dec)} with a security parameter $\lambda$ and a key space $\mathcal{K} = \{0,1\}^{Kem.Len(\lambda)}$,  is a triple of PPT algorithms defined as follows:
 
\begin{enumerate}
    \item $\text{\sf KEM.Gen}(1^{\lambda}) \rightarrow \text{(pk, sk)}$: The randomized key generation algorithm that takes $\lambda \in \mathbb{N}$ as a security parameter and returns a public and secret key pair $\text{(pk,sk)}$.
    \item $\text{\sf KEM.Enc}(\text{pk}) \rightarrow (K,C)$: This randomized algorithm takes the public key $\text{pk}$ and outputs a key $K \in \mathcal{K}$ and its ciphertext $C$.
    
    \item $\text{\sf KEM.Dec}(\text{sk}, C) \rightarrow K' \text{ or } \bot$: This deterministic algorithm takes the secret key $\text{sk}$ and the ciphertext $C$, and returns a key $K' \in \mathcal{K}$, or $\bot$.
\end{enumerate}
    
\end{definition}

\noindent \textbf{Correctness}: A KEM is $\delta(\lambda)$-correct if for all $(\text{sk}, \text{pk}) \leftarrow
\text{\sf KEM.Gen}(1^{\lambda})$ and $(K,C) \leftarrow \text{\sf KEM.Enc}(\text{pk})$, it holds that
\[
\Pr[\text{\sf KEM.Dec}(\text{sk}, C) \neq K] \leq \delta(\lambda),
\]
where probability is over the choices of $(\text{sk}, \text{pk})$ and the randomness of $\text{\sf KEM.Enc}(\cdot)$. $\delta(\cdot)$ is a non-negative negligible function in $\lambda$.

\begin{definition}[Security of KEM: IND-CPA, IND-CCA1, IND-CCA2 \cite{herranz2006kem}]
\label{def:KEM sec}
Let {\sf K} = ({\sf K.Gen}, {\sf K.Enc}, {\sf K.Dec}) be an KEM scheme with security parameter $\lambda$ and the key space $\{0,1\}^{k.Len(\lambda)}$. Let $A= (A_1,A_2)$ be an adversary. Security of a KEM {\sf K} is defined by the advantage  $Adv^{kind\text{-b}}_{{\sf K},A}(\lambda)$ of a computationally bounded adversary in the key indistinguishability (kind)   game $\text{KIND}^{\text{atk-b}}_{K,A}$,  where $b \in \{0,1\}$, shown in Figure \ref{KEM game}  below, and is given by the following expression,

{\small{
\begin{align*}
    Adv^{kind\text{-} \text{atk}}_{{\sf K},A} (\lambda)\triangleq |\Pr[\text{KIND}^{\text{atk-0}}_{{\sf K},A}(\lambda) = 1]- \Pr[\text{KIND}^{\text{atk-1}}_{{\sf K},A}(\lambda) = 1]|.
\end{align*}
}}

\begin{figure}[h!]
    \centering
    \begin{minipage}{0.5\textwidth}
\centering
\begin{lstlisting}[mathescape]
GAME $\text{KIND}^{\text{atk-b}}_{{\sf K},A}(\lambda)$:
01 $(\text{pk,sk})\stackrel{\$} \gets {\sf K.Gen}(1^\lambda)$
02 $st\stackrel{\$}\gets A_1^{{O}_1} (1^\lambda, \text{pk})$
03 $(K,C^*)\stackrel{\$}\gets {\sf K.Enc}(1^\lambda, \text{pk})$
04 (b=0) $K^* \leftarrow K$ 
04 (b=1) $\hat{K}\stackrel{\$}\gets\{0,1\}^{k.Len(\lambda)}$, $K^* \leftarrow \hat{K}$ 
05 $\textbf{return}$ $b^{\prime} \leftarrow A_2^{{O}_2}(1^\lambda, C^{*}, K^*,st)$
\end{lstlisting}
\end{minipage}
    \caption{The distinguishing game $\text{KIND}^{\text{atk-b}}_{{\sf K},A}$}
    \label{KEM game}
\end{figure}

where oracle access for each attack is defined below, and $\varepsilon$ denotes an empty string.
{
\small{
\begin{center}
\begin{tabular}{l l l}
    \text{atk} & \ $O_{1}(\cdot)$ & $O_{2}(\cdot)$ \\
    \hline
     \text{cpa}& \ $\ \varepsilon$ & $\varepsilon$\\
     \text{cca1}& \ \ ${\sf K.Dec}(\text{sk},\cdot)$ & $\varepsilon$\\
     \text{cca2}& \ \ ${\sf K.Dec}(\text{sk},\cdot)$ \ \ &${\sf K.Dec}(\text{sk},\cdot)$\\
\end{tabular}
\end{center}
}}
We say that {\sf K} provides $\sigma(\lambda) \text{-IND}\text{-ATK}$ security if for all computationally bounded adversaries $A$,  $Adv^{kind\text{-} \text{-atk}}_{{\sf K},A}(\lambda)\le \sigma(\lambda)$. $\sigma(\lambda)$ is a negligible function of $\lambda$.
\end{definition}

\begin{definition}[Data Encapsulation Mechanism (DEM)]\label{def:dem}
    A data encapsulation mechanism {\sf DEM = (DEM.Gen, DEM.Enc, DEM.Dec)} with a security parameter $\lambda$ , a key space $\mathcal{K} = \{0,1\}^{Dem.Len(\lambda)}$ and for arbitrary long messages
    is defined by the following three algorithms. 
\begin{enumerate}
     \item {\sf DEM.Gen}($1^\lambda$)$\rightarrow K$: A randomized key-generation algorithm generates a key $K$ that follows a uniform distribution in $\mathcal{K}$.
    \item ${\sf DEM.Enc}(1^\lambda, M, K)\rightarrow C$: A randomized encryption algorithm encrypts the message $M$ using the key $K$  and it produces a ciphertext $C$.
    \item ${\sf DEM.Dec}(C, K) \rightarrow M ~\text{or} \perp$: A deterministic decryption algorithm decrypts the ciphertext $C$ using the key $K$, yielding either a message $M$ or the special rejection symbol $\perp$.
\end{enumerate}

\end{definition}

\noindent We require {\em perfect correctness} which requires for any message $M \in \{0,1\}^*$, i.e.
\[\Pr[{\sf DEM.Dec}({\sf DEM.Enc}(M,K),K)=M]=1\]
where the probability is taken over $K\stackrel{\$}\gets {\sf DEM.Gen}(1^\lambda)$ and nd the randomness of all algorithms involved in the expression.

\begin{definition}[Security of DEM: IND-OT, IND-OTCCA, IND-CPA, IND-CCA1, IND-CCA2 \cite{herranz2006kem}]\label{def:demsec}
Let {\sf D}=({\sf D.Gen, D.Enc, D.Dec}) be a DEM scheme with the key space $\{0,1\}^{d.Len(\lambda)}$ and an unrestricted message space. Let $A = (A_1, A_2)$ be an adversary. Security of a DEM {\sf D} is defined by the advantage  $Adv^{ind\text{-atk}}_{{\sf D}, A}(\lambda)$ of a computationally bounded adversary in the indistinguishability (ind) game $\text{IND}^{\text{atk-b}}_{D,A}$,  where $b \in \{0,1\}$, shown in Figure \ref{DEM game}  below, and is given by the following expression,

\begin{align*}
    Adv^{ind\text{-} \text{atk}}_{{\sf D},A}(\lambda)\triangleq |\Pr[\text{IND}^{\text{atk-0}}_{{\sf D},A}(\lambda) = 1]- \Pr[\text{IND}^{\text{atk-1}}_{{\sf D},A}(\lambda) = 1]|.
\end{align*}

\begin{figure}[h!]
    \centering
    \begin{minipage}{0.4\textwidth}
\centering
\begin{lstlisting}[mathescape]
GAME $\text{IND}^{\text{atk-b}}_{{\sf D},A}(\lambda)$:
01 $K\stackrel{\$}\gets {\sf D.Gen}(1^\lambda)$
02 $(st,M_0,M_1)\stackrel{\$}\gets A_1^{{O}_1}(1^\lambda)$
03 $C^{*} \xleftarrow{\$} {\sf D.Enc}(1^\lambda, K, M_b)$
04 $\textbf{return}$ $b^{\prime} \leftarrow A_2^{{O}_2}(C^{*},st)$
\end{lstlisting}
\end{minipage}
    \caption{The distinguishing game $\text{IND}^{\text{atk-b}}_{{\sf D},A}$}
    \label{DEM game}
\end{figure}

where oracle access for each attack is defined below, and $\varepsilon$ denotes an empty string.
{
\footnotesize{
\begin{center}
\begin{tabular}{p{3em} p{11em} l}
    $\text{atk}$ &$O_{1}$ &$O_{2}$ \\
    \hline
     $\text{ot}$& $\varepsilon$ & $\varepsilon$\\
     $\text{otcca}$& $\varepsilon$&${\sf  D.Dec}(K,\cdot)$\\
     $\text{cpa}$&${\sf  D.Enc}(K,\cdot)$&$\varepsilon$\\
     $\text{cca1}$&$\{{\sf  D.Enc}(K,\cdot), {\sf  D.Dec}(K,\cdot)\}$& $\varepsilon$\\
     $\text{cca2}$&$\{{\sf  D.Enc}(K,\cdot),{\sf  D.Dec}(K,\cdot)\}$& $\{{\sf  D.Enc}(K,\cdot),{\sf  D.Dec}(K,\cdot)\}$\\

\end{tabular}
\end{center}
}}
{\sf D} provides $\sigma(\lambda) \text{-IND}\text{-ATK}$ security if for all computationally bounded adversaries $A$, $Adv^{ind\text{-}atk}_{{\sf D},A}(\lambda)\leq \sigma(\lambda)$, where $\sigma(\cdot)$ is a non-negative negligible function in $\lambda$.
\end{definition}

\begin{definition}[Hybrid Encryption in Public Key Model (HPKE)]
\label{Def:pkhybrid}
    Let a KEM {\sf K = (K.Gen, K.Enc, K.Dec)} and a DEM {\sf D = (D.Gen, D.enc, D.Dec)} share the same key space, with $\lambda$ as the security parameter. Hybrid encryption {\sf HPKE =(HPKE.Gen, HPKE.Enc, HPKE.Dec)} in the public key model for arbitrary long messages is defined by the following three algorithms. 
\begin{figure}[h!]
   \begin{center}
\begin{tabular}{l|l }
  $\mathbf{Alg}\ {\sf HPKE.Gen}(1^\lambda)$ &  $\mathbf{Alg}\ {\sf HPKE.Enc}(\text{pk,M})$ \\
      \small$(\text{pk,sk})\stackrel{\$}\gets {\sf K.Gen}(1^\lambda)$ & \small$(C_1,K)\stackrel{\$}\gets {\sf K.Enc}(\text{pk})$ \\
      \small Return $(\text{pk,sk})$ &\small $C_2\gets {\sf D.Enc}(K,M)$ \\
     &\small Return $C=(C_1,C_2)$ \\
    
\end{tabular}

\begin{tabular}{l}
\\
$\mathbf{Alg}\ {\sf HPKE.Dec}(\text{sk},C_1,C_2)$\\
\small \small $K \gets{\sf K.Dec}(\text{sk},C_1)$\\
\small \small If $K=\perp$:Return $\perp$\\ 
\small Else:\\
 \small \quad $\text{M}\gets {\sf D.Dec}(C_2,\text{K})$\\
\small Return $\text{M}$
\end{tabular}

\end{center}
    \caption{HPKE: Hybrid encryption in public key model}
    \label{fig:pkHE.def}
\end{figure}
\end{definition}

The following general composition theorem gives the security of hybrid encryption for various security notions \cite[Theorem 5.1]{herranz2006kem}.

\begin{theorem}[{\cite[Theorem 5.1]{herranz2006kem}}][\textit{IND-ATK KEM + IND-ATK$'$ DEM $\Rightarrow$ IND-ATK HPKE}]
Let $\text{ATK}\in\{\text{CPA, CCA1, CCA2}\}$ and $\text{ATK}^\prime\in\{\text{OT, OTCCA}\}$.
If $\text{ATK} \in\text{ATK} \{\text{CPA, CCA1}\}$ and $\text{ATK}^\prime$ = OT,  then HPKE, the hybrid public-key encryption scheme is a secure public-key encryption scheme under IND-ATK attack.
Similarly, if ATK = CCA2, $\text{ATK}^\prime$ = OTCCA, HPKE is a secure public-key encryption scheme under an IND-ATK attack.
\end{theorem}

\subsection{Hybrid Encryption in Preprocessing Model}
\label{HE preprocessing}

Hybrid encryption in the preprocessing model, called HE for short,  was introduced in \cite{sharifian2021information}. Instead of the public key, HE requires Alice and Bob to have samples of two correlated random variables that are partially leaked to Eve. This initial setup allows KEM to be defined with security against a computationally unbounded. This is called the information-theoretic key encapsulation mechanism, or iKEM for short.

\begin{definition}{({iKEM}) \cite{sharifian2021information} }
\label{Def:iKEM}
    An iKEM {\sf iK = (iK.Gen, iK.Enc, iK.Dec)} with security parameter $\lambda$, input distribution $P(XYZ)$ and key space $\{0,1\}^{{ iK.Len} (\lambda)}$ is defined by,
\begin{enumerate}
    \item ${\sf iK.Gen}(1^\lambda, P) \rightarrow (X, Y, Z)$: The generation algorithm generates private samples $X$, $Y$ and, $Z$ from a public probability distribution $P(XYZ)$ for Alice, Bob and, Eve respectively.
    \item ${\sf iK.Enc}(X) \rightarrow (C,K)$: The encapsulation algorithm is a probabilistic algorithm that takes Alice’s random sample $X$ as input and produces a random key $K \in \{0,1\}^{{{iK}}.Len(\lambda)}$ and a ciphertext $C$ that is used to recover K by Bob. We use the notation $ C ={\sf iK.Enc}(X).ct  $ and $ K={\sf iK.Enc}(X).key$.
    
    \item ${\sf iK.Dec}(Y,C) \rightarrow$ $K~ \text{or} \perp$: The decapsulation algorithm is a deterministic algorithm that takes the random sample $Y$ and the ciphertext $C$ as input, and recovers the key $K$ or aborts with $\perp$ as the output.
     \end{enumerate}
\end{definition}

\noindent \textbf{Correctness}: {\em An iKEM {\sf iK} is called $\delta(\lambda)$-correct} if for any samples $(X,Y)$ that is the output of {\sf iK.Gen} algorithm satisfies $\Pr[{\sf iK.Dec}(Y,C) \neq {\sf iK.Enc}(X).key] \le \delta(\lambda)$, where $\delta(\lambda)$ is a negligible function of $\lambda$, and the probability is over all randomness of ${\sf iK.Gen}$ and ${\sf iK.Enc}$ algorithms.

\noindent \textbf{Security of iKEM}: Security of iKEM is defined against a computationally unbounded adversary that has the side information $Z$ about $X$ and $Y$, and can query the {\em encapsulation} and {\em decapsulation oracles}, where each oracle implements the corresponding algorithm and has access to the corresponding private variables $X$ and $Y$, respectively.  For a security parameter $\lambda$,
we consider the following types of  adversaries:\\
--{\em OT (One-Time) adversary} with no access to  encapsulation or decapsulation
queries.\\
--{\em $q_e$-CEA (Chosen Encapsulation Adversary)} that can query the encapsulation oracle $q_e$ times, where $q_e$ is a given integer function of $\lambda$.\\
--{\em $(q_e, q_d)$-CCA 
(Chosen Ciphertext Adversary)} that can query the encapsulation and decapsulation oracles  $q_e$ and $q_d$ times, respectively, where $q_e$ and $q_d$ are given integer functions of $\lambda$.

\noindent {\em Notes:} Unlike public-key KEM, in iKEM access to encapsulation oracle is not “free” and is considered a resource. The number of queries in iKEM is a given polynomial in $\lambda$. This is in contrast to computational KEM where the number of queries is an arbitrary polynomial in $\lambda$.

\begin{definition}[Security of iKEM: IND-OT, IND-$q_e$-CEA, IND-$(q_e,q_d)$-CCA \cite{sharifian2021information}]
\label{def:iKEM sec}
Let {\sf iK} = ({\sf iK.Gen}, {\sf iK.Enc}, {\sf iK.Dec}) be an KEM scheme with security parameter $\lambda$, input distribution $P(XYZ)$ and the key space $\{0,1\}^{{ik.Len}(\lambda)}$. Security of an iKEM {\sf iK} is defined by the advantage  $Adv^{ikind\text{-atk}}_{{\sf iK},A}(\lambda)$ of a computationally unbounded adversary in the information-theoretic key indistinguishability (ikind) game $\text{IKIND}^{\text{atk-b}}_{{\sf{iK},A}}$,  where $b \in \{0,1\}$, shown in Figure \ref{iKEM game}  below, and is given by the following expression,
{\small{
\begin{align*}
    Adv^{ikind\text{-} \text{atk}}_{{\sf iK},A} (\lambda)\triangleq |\Pr[\text{IKIND}^{\text{atk-0}}_{{\sf iK},A}(\lambda) = 1]- \Pr[\text{IKIND}^{\text{atk-1}}_{{\sf iK},A}(\lambda) = 1]|.
\end{align*}
}
}
\begin{figure}[h!]
    \centering
    \begin{minipage}{0.45\textwidth}
\centering
\begin{lstlisting}[mathescape]
GAME $\text{IKIND}^{\text{atk-b}}_{{\sf iK},A}(\lambda)$:
01 $(\text{X,Y,Z})\stackrel{\$} \gets {\sf iK.Gen}(1^\lambda, P)$
02 $st\stackrel{\$}\gets A_1^{{O}_1} (Z)$
03 $(K,C^*)\stackrel{\$}\gets {\sf iK.Enc}(1^\lambda, X)$
04 (b=0) $K^* \leftarrow K$ 
04 (b=1) $K_1\stackrel{\$}\gets\{0,1\}^{{iK.Len}(\lambda)}$,$K^* \leftarrow \hat{K}$ 
05 $\textbf{return}$ $b^{\prime} \leftarrow A_2^{{O}_2}(C^{*}, K^*,st)$
\end{lstlisting}
\end{minipage}
    \caption{The distinguishing game $\text{IKIND}^{\text{atk-b}}_{{\sf iK},A}$}
    \label{iKEM game}
\end{figure}
$(K^*, C^*)$ is the challenge key and ciphertext pairs that are generated by the challenger in the above game using the encapsulation algorithm. Oracle accesses for each attack are defined below. $\varepsilon$ denotes an empty string.
{
\scriptsize{ 
\begin{center}
\begin{tabular}{l l l}
    \text{atk} & \textbf{$O_{1}(\cdot)$} & \textbf{$O_{2}(\cdot)$} \\
    \hline
    \text{ot} & $\varepsilon$ & $\varepsilon$ \\
    $q_e$\text{-cea} & ${\sf iK.Enc}(X,\cdot)$ & $\varepsilon$ \\
    $(q_e,q_d)$\text{-cca} & 
    $\{{\sf iK.Enc}(X,\cdot),{\sf iK.Dec}(Y,\cdot)\}$ & 
    $\{{\sf iK.Enc}(X,\cdot),{\sf iK.Dec}(Y,\cdot)\}$ \\
\end{tabular}
\end{center}
}
}

We say that {\sf iK} provides $\sigma(\lambda) \text{-IND}\text{-ATK}$ security if 
$Adv^{ikind\text{-} \text{-atk}}_{{\sf iK},A}(\lambda)\le \sigma(\lambda)$ for all computationally bounded adversaries $A$ where  $atk \in\{\text{ot},  \text{$q_e$-cea},  \text{$(q_e,q_d)$-cca}\}$, and $\sigma(\lambda)$ is a negligible function of $\lambda$. 

\end{definition}

The following 
lemma \cite[Lemma 3]{sharifian2021information} shows that the distinguishing advantage of the adversary in definition  \ref{def:iKEM sec} is bounded by the statistical distance  of the $K^*$ and uniform distribution, given the adversary's view.


\begin{lemma}[{\cite[Lemma 3]{sharifian2021information}}]
\label{IND-CEA iKEM}
    An iKEM {\sf iK} is $\sigma(\lambda)$-IND-$q_e$-CEA secure if and only if for all computationally unbounded adversary,
 \begin{equation}
        SD((Z, C^{*}, K^{*}), (Z, C^{*}, U_{iK.Len(\lambda)}))\le \sigma(\lambda)
    \end{equation}
Here $U_{{iK}.Len(\lambda)}$ is the uniform distribution over $\{0,1\}^{{ iK}.Len(\lambda)}$ and $\sigma(\lambda)$ is a negligible function of $\lambda$.
\end{lemma}

\noindent DEM security in the preprocessing model is defined the same as in the public-key setting (Definition \ref{def:demsec}), assuming a {\em computationally bounded} adversary.

\begin{definition}[(Hybrid Encryption (HE) in Preprocessing Model) \cite{sharifian2021information}] Let {\sf iK} = ({\sf iK.Gen, iK.Enc, iK.Dec}) and {\sf D = (D.Gen, D.Enc, D.Dec)} represent an information-theoretic key encapsulation mechanism (iKEM) and a data encapsulation mechanism (DEM), respectively, both with security parameter $\lambda$ and with the same key space $\{0,1\}^{{HE.Len}(\lambda)}$. For a public  distribution $P(XYZ)$, the hybrid encryption in preprocessing model, denoted by ${\sf HE_{iK,D}} = ({\sf HE.Gen}, {\sf HE.Enc},\\{\sf HE.Dec} )$ is given in Figure \ref{fig:HE information}, where the message space is an unrestricted message space $\{0, 1\}^*$ for all $\lambda$.

\begin{figure}[h!]
   \begin{center}
\begin{tabular}{l| l }
  $\mathbf{Alg}\ {\sf HE.Gen}(1^\lambda, P)$ &  $\mathbf{Alg}\ {\sf HE.Enc}(X,M)$\\
      \small$(X,Y,Z)\stackrel{\$}\gets {\sf iK.Gen}(1^{\lambda}, P)$&\small$(C_1,K)\stackrel{\$}\gets {\sf iK.Enc}(X)$\\
      \small Return $(X,Y,Z)$ &\small $C_2\gets {\sf D.Enc}(K,M)$\\
     &\small Return $(C_1,C_2)$
\end{tabular}
\begin{tabular}{l}
\\
$\mathbf{Alg}\ {\sf HE.Dec}(Y,C_1,C_2)$\\
\small $K \gets{\sf iK.Dec}(Y,C_1)$\\
\small If $K=\perp$:Return $\perp$\\ 
\small Else:\\
\small \quad $M\gets {\sf D.Dec}(C_2,K)$\\
\small Return $M$
\end{tabular}
\end{center}
    \caption{Hybrid encryption in preprocessing model}
    \label{fig:HE information}
\end{figure}
    
\end{definition}

\noindent \textbf{Security of HE}: Security of HE is defined against OT, $q_e$-CPA and $(q_e,q_d)$-CCA adversaries.

\begin{definition}[(Security of Hybrid Encryption in Preprocessing Model]
\label{def: pHE-sec}
For a security parameter, $\lambda$ let {\sf HE$_{iK,D}$ = (HE.Gen, HE.Enc, HE.Dec)} be a hybrid encryption scheme obtained from an iKEM {\sf iK = (iK.Gen, iK.Enc, iK.Dec)} defined for a public distribution $P(XYZ)$ and a DEM {\sf D = (D.Gen, D.Enc, D.Dec)}. Security of ${\sf HE_{iK,D}}$ is formalized by bounding the advantage $Adv^{ind\text{-} \text{atk}}_{{\sf HE},A} (\lambda)$ of a computationally bounded adversary $A=(A_1, A_2)$ in the indistinguishability game $\text{IND}^{\text{atk-b}}_{{\sf HE},A}$, $b \in \{0, 1\}$, shown in Figure \ref{pHE game}, defined by the expression,
{\small{\begin{align*}
Adv^{ind\text{-} \text{atk}}_{{\sf HE},A} (\lambda)\triangleq |\Pr[\text{IND}^{\text{atk-0}}_{{\sf HE},A}(\lambda) = 1]- \Pr[\text{IND}^{\text{atk-1}}_{{\sf HE},A}(\lambda) = 1]|
\end{align*}}}

where the adversary's capabilities are given by,
{
\scriptsize{ 
\begin{center}
\begin{tabular}{l l l}
    $\text{atk}$ & $O_{1}$ & $O_{2}$ \\
    \hline
    $\text{ot}$ & $\varepsilon$ & $\varepsilon$ \\
    $q_e\text{-cpa}$ & ${\sf HE.Enc}(X,\cdot)$ & $\varepsilon$ \\
    $(q_e,q_d)\text{-cca}$ & $\{{\sf HE.Enc}(X,\cdot), {\sf HE.Dec}(Y,\cdot)\}$ & $\{{\sf HE.Enc}(X,\cdot), {\sf HE.Dec}(Y,\cdot)\}$ \\
\end{tabular}
\end{center}
}}

\begin{figure}[h!]
    \centering
    \begin{minipage}{0.45\textwidth}
\centering
\begin{lstlisting}[mathescape]
GAME $\text{IND}^{\text{atk-b}}_{{\sf HE},A}(\lambda)$:
01 $(X, Y, Z)\stackrel{\$} \gets {\sf HE.Gen}(\lambda,P)$
02 $(st, M_0,M_1)\stackrel{\$}\gets A_1^{O_1} (Z)$
03 $C^*\stackrel{\$}\gets {\sf HE.Enc}(X, M_b)$
04 $\textbf{return}$ $b^{\prime} \leftarrow A_2^{O_2}(C^{*}, K_b,st)$
\end{lstlisting}
\end{minipage}
    \caption{The distinguishing game $\text{IND}^{\text{atk-b}}_{{\sf HE},A}$}
    \label{pHE game}
\end{figure}

We say that ${\sf HE_{iK,D}}$ provides $\sigma(\lambda)$-IND-ATK  security if 
$Adv^{ind\text{-} \text{atk}}_{{\sf HE},A} (\lambda)\le \sigma(\lambda)$  for all computationally bounded adversaries $A$ where   $\text{atk} \in \{  \text{ot},  \text{$q_e$-cpa},  \text{$(q_e,q_d)$-cca}\}$, and $\sigma(\lambda)$ is a negligible function of $\lambda$.

\end{definition}
\noindent\textbf{
Security Theorem for HE}: Security of HE against a computationally bounded adversary with  access to $q_e$ encapsulation  queries, constructed from an iKEM and a computationally secure DEM, was  proved in  the following theorem.

\begin{theorem}[{\cite[Theorem 2]{sharifian2021information}}]
\label{Composable security}
   Let an iKEM ${\sf iK}$ be $\sigma(\lambda)$-IND-$q_e$-CEA secure (information-theoretically secure) and a DEM ${\sf D}$ be $\sigma'(\lambda)$-IND-OT secure (computationally secure). Then, the corresponding hybrid encryption scheme ${\sf HE}_{{\sf iK}, {\sf D}}$ achieves $[\sigma(\lambda) + \sigma'(\lambda)]$-IND-$q_e$-CPA security against a computationally bounded adversary $A$.
\end{theorem}

\section{Quantum Key Distribution}
\label{QKD} 
There is a large body of research on QKD - see \cite{nurhadi2018quantum,portmann2022security, wolf2021quantum} and references therein. We consider BB84 style QKDs, also known as {\em prepare-and-measure} QKD.

In the following we will describe the QKD protocol and its algorithms as given in \cite[Section V]{portmann2022security}. We use this QKD in section \ref{instantiation:qKEM} to construct a qKEM. Before the protocol starts, Alice (A) and Bob (B) agree on two sets of basis vectors $\mathbb{X} = \{\ket{0}, \ket{1}\}$ and $\mathbb{Z} = \{\ket{+}, \ket{-}\}$ where $\ket{\pm} = \frac{1}{\sqrt{2}}(\ket{0} \pm \ket{1})$. The protocol proceeds as follows.\\
-- \textbf{Prepare (and send)}: For a security parameter $\lambda$, Alice generates a $\lambda$-bit random string, encodes each bit into a quantum state using a randomly selected basis from the set $\{\mathbb{X}, \mathbb{Z}\}$, and sends the encoded quantum state $\rho_A$ to Bob through an insecure quantum channel.\\
-- Eve interacts with the quantum channel, producing a joint state $\rho = U (\rho_A \oplus \ket{0}\bra{0}_E) U^\dagger$ by attaching a quantum state $\ket{0}\bra{0}_E$ to $\rho_A$ using a joint unitary operation U. The respective reduced density matrices of Bob and Eve are $\rho_B = \text{Tr}_E(\rho)$ and $\rho_E = \text{Tr}_B(\rho)$. Eve passes $\rho_B$ to Bob. Eve  may measure her system at any time during the protocol i.e. she may delay measurement until the protocol is completed (Alice and Bob obtain their final Keys).\\
-- \textbf{Measure}: Bob measures the received sequence of quantum states using a randomly chosen sequence of basis from $\{\mathbb{X}, \mathbb{Z}\}$ and obtains an initial binary sequence.\\
-- \textbf{Sifting}: Alice and Bob communicate over a classical authenticated channel, following  Algorithm \ref{State prep and Sifting}, to {\em sift} their sequence and obtain their individual  $N$-bit sequences with the guarantee that they have used the same basis for their respective (Alice’s prepare, and Bob’s measure)  operations. \\
-- \textbf{Parameter Estimation}: Alice and Bob use Algorithm \ref{p.estimation} to measure the Quantum Bit Error Rate (QBER) of their bit strings by selecting a random sample of size $r$. Note that because of noise in any real-life quantum channel, QBER is always non-zero. The protocol aborts if QBER is above a predetermined threshold $\eta_0$; else Alice and Bob obtain raw keys ${\bf X}_A$ and ${\bf X}_B$ of length $n$ where $n = N-r$ (provided every signal is detected by Bob).\\
-- \textbf{Reconcilliation}: Alice and Bob use Algorithm \ref{I.reconci} to remove the errors from their respective strings and share a common string and have ${\bf X}_A$ and ${\bf X}_B^{\prime}$ as shared key.\\
-- \textbf{Privacy amplification}: Alice and Bob use Algorithm \ref{P.Amplification} to extract, respectively, the final bit-string ${\bf K}_A$ and ${\bf K}_B$ of length $\ell$.

\noindent \textbf{Classical-quantum (cq) state}: The {\em joint classical-quantum} (cq) state for 
 a classical random variable $K$ with domain $\mathcal{K}= \{0,1\}^\ell$, and the quantum system $E$ which can encompass both quantum and classical information and is specified by  with density matrix $\rho^{{\bf k}}_E$, is defined as,

\begin{equation*}
    \rho_{{\bf{K}}_A E}= \sum_{{\bf k} \in \mathcal{K} \cup \{\perp\}} \Pr({\bf k}) \ket{{\bf k}}\bra{{\bf k}} \otimes \rho^{{\bf k}}_E  
\end{equation*}
where $\Pr ({\bf k}) = \Pr ({\bf K}_A={\bf k})$ and $\{\ket{{\bf k}}\}_{{\bf k} \in \mathcal{K} \cup \{\perp\}}$ is a family of orthonormal vectors in some Hilbert Space.

\subsection{Post-processing}
Algorithms \ref{p.estimation}, \ref{I.reconci}, \ref{P.Amplification} are run after the sifting step. The {\em parameter estimation protocol} (Algorithm \ref{p.estimation}) estimates the error rate $\eta$ and returns ``fail"  if it is above $\eta_0$.

In \cite[Equation 22]{portmann2022security},
using the result in \cite{renner2008security} that the {\em smooth min-entropy} of ${\bf X}_A$ conditioned on Eve's information $E$, denoted by $H^{\varepsilon}_{\text{min}}({\bf X}_A|E) $, gives the minimum number of bits that can be extracted from $\bf X_A$ such that the bits are uniformly distributed and are uncorrelated  with $ E$ with an error probability $\varepsilon$, it is shown that,

\begin{equation}
\label{min entropy}
    H^{\varepsilon}_{\text{min}}({\bf X}_A|E) \ge n\cdot [1 - h(\eta_0)] +  O(\sqrt{n})
\end{equation}

where  $h(x) = -x\log_2(x) -(1-x)\log_2(1-x)$ is the
binary entropy function. This result assumes collective attack.

The {\em information reconciliation} (Algorithm \ref{I.reconci})  is a purely classical protocol. It uses a pair of $(enc, dec)$ functions that is designed for a specific noise model that is captured by a joint distribution on $({\bf X}_A, {\bf X}_B)$ and is constructed such that 
\begin{equation}\label{rec}
\Pr(dec(enc({\bf X}_A), {\bf X}_B) = {\bf X}_A )\geq 1- \varepsilon 
\end{equation}
To achieve reconciliation, Alice must send some information ${\bf W}= enc({\bf X})$ to Bob. It was proved that \cite{renner2005simple} that achieving \ref{rec} requires ${\bf W}$ to be at least $H^\varepsilon_{max}({\bf X}_A| {\bf X}_B)$ bits, where $H^\varepsilon_{max}$ is the {\em smooth max entropy}, and the bound is achievable for some coding schemes. Sending reconciliation information further reduces the min-entropy of ${\bf X}_A$ and so shortens the length of the extracted key. It was proved that for the case of BB84 protocol \cite{portmann2022security}, 

\begin{equation}
H^\varepsilon_{min} ({\bf X}_A| E{\bf W}) \geq n[1-2h(\eta_0)]-O(\sqrt{n})
\end{equation}

{\em Privacy amplification} (Algorithm \ref{P.Amplification}) transforms a partially leaked secret key that, after information reconciliation, is shared between Alice and Bob into a strong (almost) uniformly distributed secret key ${\bf K}_A$ (or ${\bf K}_B$) that is independent of any information accessible to Eve. In the context of QKD, the most commonly used extractor is two-universal hash functions (2-UHF) \cite{carter1977universal} (see Definition \ref{2-UHF}). The extracted key has length $\ell$, where 
\begin{equation}
   \ell =  H^\varepsilon_{min} ({\bf X}_A| E{\bf W}) - O(1)
\end{equation}

\subsection{Security}
Eve is assumed to have unlimited resources and be constrained only by the laws of physics. 

\subsubsection{ Attack strategies} Eve's interaction strategy with the quantum communication system has been divided into {\em individual attack}, {\em collective attack} and {\em coherent attack} \cite{renner2005information}. In an {\em individual attack}, Eve interacts with every transmitted qubit, that is sent by Alice, independently. In a {\em collective attack}, similar to the individual attack, Eve interacts with each qubit independently but delays its measurement until after the completion of the privacy amplification step. In a {\em coherent attack}, the most powerful
attack strategy available to Eve, Eve attaches a single (large) ancilla system to all quantum states (qubits)
that are transmitted by Alice, apply quantum processing  to the entire composite system, and makes measurement at the completion of the QKD protocol. 

To prove security against a coherent attacker, a general approach is to first prove security against a collective attacker, and then ``lift'' the proof to provide security against a coherent attacker \cite{renner2008security, sheridan2010finite, scarani2008quantum}.

\begin{algorithm}[h!]
\caption{State Preparation, transmission and Sifting}\label{State prep and Sifting}
\begin{algorithmic}[1]
\STATE \textbf{Parameters:} $\lambda$ (number of signals); Bases $\{\mathbb{X},\mathbb{Z}\}$ (encoding)
\STATE $i \gets 1$
\WHILE{$i \leq \lambda$}
    \STATE Alice: \\
    - Chooses a basis $B_i \in_R \{\mathbb{X},\mathbb{Z}\}$ and a bit $X_{(A,i)} \in_R \{0,1\}$\\
    - Encodes the bit $X_{(A,i)}$ using the basis $B_i$ in quantum state $\rho_A^{i}$ and sends to Bob through a quantum channel
    \STATE Bob: \\
    - Receives a quantum state $\rho_B^{i}$\\
    - Selects a basis $B_i^{\prime} \in_R \{\mathbb{X},\mathbb{Z}\}$ and measures the received quantum state using the chosen basis\\
    - Gets an output bit $X_{(B,i)} \in \{0,1\}$ 
   
    \STATE Alice and Bob broadcasts their basis choices~~ // {\em sifting}
    \STATE They discards the corresponding bit for $B_i \neq B_i^{\prime}$
   
    \IF{$B_i = B'_i$}
        \STATE accept the corresponding bit
        \STATE $i \gets i + 1$
    \ENDIF
\ENDWHILE
\STATE \textbf{return} ${\bf X}_A^N = (X_{(A,1)}, \ldots, X_{(A,N)})$ and ${\bf X}_B^N = (X_{(B,1)}, \ldots, X_{(B,N)})$, where N is the number of accepted bits
\end{algorithmic}
\end{algorithm}


\begin{algorithm}[h!]
\caption{Parameter Estimation (${\bf X}_A^N$, ${\bf X}_B^N$)}\label{p.estimation}
\begin{algorithmic}[1]
\STATE \textbf{Parameters:} $r$ (sample size), $\eta_0$ (threshold)
\STATE Alice randomly chooses a subset $\mathcal{R}$ of length $r$ from the set $\{1,\ldots, N\}$
\STATE Alice sends the set $\{(i,X_{(A,i)}): i \in S\}$ to Bob over public authenticated channel
\STATE Alice and Bob computes the error rate $\eta = \frac{1}{|\mathcal{R}|} \sum_{i\in \mathcal{R}} |X_{(A,i)} - X_{(B,i)}|$
\IF{$\eta \leq \eta_0$}
    \STATE ok
    
\ELSE
    \STATE fail and \textbf{return} $\perp$ 
\ENDIF
\STATE \textbf{return} $n$ $(= N-r)$ bit strings ${\bf X}_A $ and ${\bf X}_B $
    
\end{algorithmic}
\end{algorithm}


\begin{algorithm}[h!]
\caption{Information Reconciliation (${\bf X}_A$, ${\bf X}_B$)}\label{I.reconci}
\begin{algorithmic}[1]
\STATE \textbf{Parameters:} $(enc,dec)$ (coding scheme)
\STATE Alice sends ${\bf W} = enc({\bf X}_A)$ to Bob over classical channel
\STATE Bob computes ${\bf X}_B^\prime = dec({\bf W}, {\bf X}_B)$
\STATE \textbf{return} (${\bf X}_A$, ${\bf X}_B^{\prime}$)
    
\end{algorithmic}
\end{algorithm}

\begin{algorithm}[h!]
\caption{Privacy Amplification (${\bf X}_A$, ${\bf X}_B^{\prime}$)}\label{P.Amplification}
\begin{algorithmic}[1]
\STATE \textbf{Parameters:} $\{h_s\}_{s \in \mathcal{S}}$ (seeded randomness extractor)
\STATE Alice:\\
- Choose $S \in_R \mathcal{S}$ and send $S$ over a authenticated classical channel\\
- Computes ${\bf K}_A = h_S({\bf X}_A)$
\STATE Bob computes ${\bf K}_B = h_S ({\bf X}_B^\prime)$
\STATE \textbf{return} (${\bf K}_A$, ${\bf K}_B$)
\end{algorithmic}
\end{algorithm}

{\bf Security of QKD} has been argued and proved in a long line of research papers \cite{lo1999unconditional,shor2000simple,mayers2001unconditional,gottesman2003proof,renner2005information, inamori2007unconditional, renner2008security, scarani2009security, hayashi2006practical}. More recent security models use  the security criterion that is known as {\em trace distance} and was introduced in  \cite{ben2005universal,konig2007small,renner2008security,hayashi2012concise}. The criterion requires that  $\rho_{{\bf K_A}E}$, the joint state of the final key $\bf K_A$ and the quantum information that has been obtained by an eavesdropper $E$ (defined below), must be close to an ideal
key $\bf K$ that is uniform and independent from the
adversary’s information $\rho_E$.

\begin{definition}[Secrecy]
\label{Def:QKD Secret}
   {\em A QKD is called $\epsilon (\lambda)$-secret} if,  
 
 \begin{equation}
        \frac{1}{2} \norm{\rho_{{\bf{K}}_A E}-\omega_{{\bf{K}}} \otimes \rho_E} _1\leq \epsilon(\lambda).
        \end{equation}
     where $\omega_{{\bf{K}}} = \sum_{{\bf k} \in \mathcal{K}} \frac{1}{\mathcal{K}} \ket{{\bf k} }\bra{{\bf k} }$ is the maximally mixed state over $\mathcal{K}$. 
\end{definition}

A QKD scheme also ensures that,  with high probability, Alice and Bob have the same key. This is called  
correctness and defined below.

\noindent \textbf{Correctness}: A QKD is called $\delta(\lambda)$-{correct} if,

\begin{equation*}
    \Pr[{\bf K}_A \neq {\bf K}_B] \le \delta(\lambda)
\end{equation*}

where the probability is over the randomness of all parties, including the adversary, in the protocol.

\noindent \textbf{Security}: {\em A QKD protocol
that is $\delta(\lambda)$-correct and $\epsilon(\lambda)$-secret is called $\sigma(\lambda)$-secure} if it satisfies $\delta(\lambda) + \epsilon(\lambda) \leq \sigma(\lambda)$.

 Finally, {\em robustness} of QKD schemes requires that under reasonably noisy conditions the protocol, with high probability,   produce a key.  In this paper we assume perfect robustness; that is the protocol always generates a key.

\section{Quantum-enabled KEM}
\label{qKEM}
We aim to use KEM/DEM in preprocessing model framework to define and prove computational  security of encrypting a message from an unrestricted domain using a QKD (defined in Section \ref{QKD}) derived key and a secure OT symmetric key encryption system. Towards this goal we define a quantum-enabled KEM, or QKEM for short, as a tuple of algorithms (see Definition \ref{Def:qKEM}) that are used by Alice and Bob in presence of Eve. The entities and their capabilities are defined below.

\subsection{Setting, Entities and Defintions} 
\label{adversary}
Alice and Bob are connected by an {\em insecure quantum channel} and a public authenticated channel for transmission of classical data.
They hold quantum devices that can perform ``prepare and measure'' operation that is used for encoding of classical information into quantum state, and measurement (decoding) of quantum encoded 
data to classical information, as required in the prepare-and-measure step of BB84-like protocols. We assume that Alice's source and Bob's detector are ideal and encoding and decoding of classical data does not generate any error. Other than prepare-and-measure capability, Alice and Bob {\em do not have other} quantum computation capabilities and use classical computation only. 

{\em Eve is  a quantum adversary} that has  full control over the quantum channel, can observe classical information that are passed over the classical channel, has {\em unlimited quantum  (computation) resources} and its strategies are only  bounded by the law of quantum mechanics. Alice and Bob use their channels to obtain private  correlated samples   $X$ and $Y$, and an estimate of $Z$, Eve's information about their  variables. We refer to a KEM in pre-processing model 
in this setting
as {\em  quantum-enabled KEM, or qKEM for short}. Compared to the iKEM (Definition \ref{Def:iKEM}) the main difference is that the distribution $\cal P$ is not assumed known and input to the correlation generation algorithm. Rather, it is generated by the interaction of Alice and Bob in the presence of Eve, whose information can be estimated by Alice and Bob through measuring errors in quantum channel. We use  $(.{\sf Gen})_\mathcal{E}$  to emphasize that this step is in presence of Eve ($\mathcal{E}$).

\begin{definition}[Quantum-enabled KEM ({qKEM})]
\label{Def:qKEM}
    A qKEM {\sf qK = (qK.Gen$_\mathcal{E}$, qK.Enc, qK.Dec)} with security parameter $\lambda$ and key space $\{0,1\}^{{ qK}.Len (\lambda)}$ is defined as follows.
\begin{enumerate}
    \item  
{\sf qK.Gen}$_\mathcal{E}(1^\lambda$) $\rightarrow (X, Y, \rho_E$): 
    The pre-processing step is run in a presence of Eve ($\mathcal{E}$), who can interact with the quantum channel and eavesdrop on all communications over the (authenticated) classical channel. Eve  can use any  attack strategy, outlined in section \ref{QKD}. Alice and Bob 
    use interaction over classical
    channel to estimate the error rate over the quantum channel.  
    If the error rate is above a threshold, the protocol is aborted; else Alice and Bob will generate correlated binary strings $X$ and $Y$ (known as raw keys in the BB84 protocol). At the end of this step, Eve will have a quantum state $\rho_E$, representing its information.
   \item {\sf qK.Enc}$(X) \rightarrow (C, K)$: This is a probabilistic algorithm that takes Alice's private string $X$ as input and produces a ciphertext $C$ and a key $K \in \{0,1\}^{{qK}.Len(\lambda)}$. We denote $ C ={\sf qK.Enc}(X).ct  $ and $ K={\sf qK.Enc}(X).key$.
    \item {\sf qK.Dec}$(Y,C) \rightarrow$ $(K~\text{or} \perp)$: This is a deterministic algorithm 
     that takes Bob's private string $Y$ and ciphertext $C$ as input, and outputs a key $K$, or  aborts and outputs $\perp$.
    
     \end{enumerate}
\end{definition}

\noindent \textbf{Correctness (Soundness)} [Section \ref{HE preprocessing} ]: A qKEM protocol ${\sf qK}$ is called $\delta(\lambda)$-correct (sound \cite{cramer2003design}) if for a pair of private samples $(X,Y)$ that is generated in Step 1 (Definition \ref{Def:qKEM}), we have,
$\Pr[{\sf qK.Dec}(Y,C) \neq K]$ $\le \delta(\lambda)$ where ${\sf qK.Enc}(X) =(K,C)$,  and $\delta(\lambda)$ is a negligible function of $\lambda$.  The probability is over the random choices of ${\sf qK.Gen}_\mathcal{E}$ and {\sf qK.Enc} algorithms. \\

\noindent \textbf{Security of qKEM}: A qKEM, when not aborted,  generates a key with information-theoretic security. We define One-Time (OT) security of qKEM (adversary without  any oracle access) using the security definition of iKEM  (Definition \ref{def:iKEM sec}). We assume qKEM always succeeds and produces keys.

\begin{definition}{(Security of  qKEM: IND-OT)}\\
\label{def:qEMOT}
Let {\sf qK} = ({\sf qK.Gen}$_\mathcal{E}$, {\sf qK.Enc}, {\sf qK.Dec}) be a qKEM scheme with security parameter $\lambda$ and the key space $\{0,1\}^{{\sf qK}.Len(\lambda)}$. Let $\mathcal{E}$ be a quantum adversary as defined in subsection \ref{adversary}. Security of 
{\sf qK} is defined by the advantage  $Adv^{ikind\text{-}ot}_{{\sf qK},\mathcal{E}}(\lambda)$ of $\mathcal{E}$ in the game $\text{IKIND}^{\text{ot-b}}_{{\sf qK},\mathcal{E}}$ in Figure \ref{qKEM game} given by the following expression, and  $b \in \{0,1\}$.

{\small{\begin{equation*}
Adv^{ikind\text{-}ot}_{{\sf qK},\mathcal{E}}(\lambda)\triangleq |\Pr[\text{IKIND}^{\text{ot-0}}_{{\sf qK},\mathcal{E}}(\lambda) = 1]- \Pr[\text{IKIND}^{\text{ot-1}}_{{\sf qK},\mathcal{E}}(\lambda) = 1]|.
\end{equation*}}}

\begin{figure}[h!]
    \centering
    \begin{minipage}{0.45\textwidth}
\centering
\begin{lstlisting}[mathescape]
GAME $\text{IKIND}^{\text{ot-b}}_{{\sf qK},\mathcal{E}}(\lambda)$:
01 $(X, Y, \rho_E)\stackrel{\$} \gets {\sf qK.Gen}_\mathcal{E} (1^\lambda)$
02 $(K,C^*)\stackrel{\$}\gets {\sf qK.Enc}(X)$
03 (b=0) $K^* \leftarrow K$ 
04 (b=1) $\hat{K}\stackrel{\$}\gets\{0,1\}^{qK.Len(\lambda)}$, $K^* \leftarrow \hat{K}$ 
05 $\textbf{return}$ $b^{\prime} \leftarrow \mathcal{E}(C^{*}, K^*, \rho_E)$
\end{lstlisting}
\end{minipage}
    \caption{The distinguishing game $\text{IKIND}^{\text{ot-b}}_{{\sf qK},\mathcal{E}}$}
    \label{qKEM game}
\end{figure}

A {\em qK is $\sigma(\lambda)$-IND-OT secure} if for all successful executions of the protocol (not aborted), the advantage $Adv^{ikind\text{-}ot}_{{\sf qK},\mathcal{E}}(\lambda)\le \sigma(\lambda)$ where $\mathcal{E}$ is an unbounded quantum adversary and $\sigma(\lambda)$ is a non-negative negligible function of $\lambda$.

\end{definition}

\section{Quantum Enabled Hybrid Encryption (qHE)}
\label{qHE}
{\em Quantum Enabled Hybrid Encryption} ({\sf qHE}) 
is a HE in pre-processing model \cite{sharifian2021information} where the $\sf Gen$ algorithm of qKEM is used to generate the private samples of Alice and Bob and enable them to estimate the correlation between their samples in terms of the Hamming distance between their samples, as well as estimating  Eve’s information in terms of QBER.

\begin{definition}[{{qHE}}]
    Let {\sf qK = (qK.Gen$_\mathcal{E}$, qK.Enc, qK.Dec)} and     {\sf D = (D.Gen, D.enc, D.Dec)}, respectively, be a qKEM and DEM with security parameter $\lambda$, having the same key space. A {\em quantum enabled hybrid encryption} is a tuple of algorithms  {\sf qHE =(qHE.Gen$_\mathcal{E}$, qHE.Enc, qHE.Dec)} that are defined in Figure \ref{fig:qHE.def}, using algorithms of qKEM and  DEM. 

\begin{figure}[h!]
\centering
  
\begin{tabular}{l| l }
  $\mathbf{Alg}\ {\sf qHE.Gen}_\mathcal{E}(1^\lambda)$ &  $\mathbf{Alg}\ {\sf HE.Enc}(X,M)$\\
      \small$(X, Y ,\rho_E)\stackrel{\$}\gets {\sf qK.Gen}_\mathcal{E}(1^\lambda)$ &\small$(C_1,K)\stackrel{\$}\gets {\sf qK.Enc}(X)$\\
      \small Return $(X,Y,Z)$ &\small $C_2\gets {\sf D.Enc}(K,M)$\\
     &\small Return $(C_1,C_2)$
\end{tabular}
\begin{tabular}{l}
\\
$\mathbf{Alg}\ {\sf HE.Dec}(Y,C_1,C_2)$\\
\small $K \gets{\sf qK.Dec}(Y,C_1)$\\
\small If $K=\perp$:Return $\perp$\\ 
\small Else:\\
\small \quad $M\gets {\sf D.Dec}(C_2,K)$\\
\small Return $M$
\end{tabular}

    \caption{\small{Quantum enabled hybrid encryption. We use subscript $\mathcal{E}$ for {\sf qHE.Gen} to emphasize that the generation algorithm is in presence of QPT adversary $\mathcal{E}$.}}
    \label{fig:qHE.def}
\end{figure}

\end{definition}

$C_1$ is the view of $\mathcal{E}$ after execution of ${\sf qK.Gen}$ and ${\sf qK.Enc}$.
\noindent \textbf{Correctness}: A qHE protocol ${\sf qHE}$ is $\delta(\lambda)$-correct if for a pair of private samples $(X,Y)$ that is generated through ${\sf qK.Gen}_\mathcal{E}$, we have
$\Pr[{\sf qHE.Dec}(Y,C) \neq M]$ $\le \delta(\lambda)$ where ${\sf qHE.Enc}(X,M) = C$, and $\delta(\lambda)$ is a negligible function of $\lambda$.  The probability is over  the random choices of ${\sf qHE.Gen}_\mathcal{E}$ and {\sf qHE.Enc} algorithms.

\noindent \textbf{Security}: We consider {\em one-time  indistinguishability security (IND-OT) of {\sf qHE}.} Alice's and  Bob's  capabilities are as defined for qKEM in section \ref{adversary}. {\em Eve, however, is a quantum adversary with polynomially bounded computational resources (QPT).} Thus security of qHE is against a QPT adversary that can run polynomially bounded algorithms on a quantum computer.

\begin{definition}{(IND-OT Security of 
{\sf qHE} 
)}\\
\label{Def:qHE OT}
For a security parameter $\lambda$, security of {\sf qHE} is defined by the advantage $Adv_{{\sf qHE}, \mathcal{E}}^{ind\text{-}ot}$ of a {\sf qHE} adversary $\mathcal{E}$ in the game $\text{IND}^{\text{ot-b}}_{{\sf qHE},\mathcal{E}}$ in Figure \ref{qHE-OT:game} where $b \in  \{0,1\}$ defined by,
{\small{\begin{align*}
    Adv^{ind\text{-} \text{ot}}_{{\sf qHE},\mathcal{E}} (\lambda)\triangleq |\Pr[\text{IND}^{\text{ot-0}}_{{\sf qHE},\mathcal{E}}(\lambda) = 1]- \Pr[\text{IND}^{\text{ot-1}}_{{\sf qHE},\mathcal{E}}(\lambda) = 1]|.
\end{align*}}}

\begin{figure}
    \centering
    \begin{minipage}{0.45\textwidth}
\centering
\begin{lstlisting}[mathescape]
GAME $\text{IND}^{\text{ot-b}}_{{\sf qHE},\mathcal{E}}(\lambda)$:
01 $(X, Y, \rho_{E})\stackrel{\$} \gets {\sf qHE.Gen}_\mathcal{E}(1^\lambda)$
02 $(M_0,M_1)\stackrel{\$}\gets \mathcal{E}(\rho_{E})$
03 $C^* = (C_1^{*},C_2^{*})\stackrel{\$}\gets {\sf qHE.Enc}(X,M_b)$
04 $\textbf{return}$ $b^{\prime} \leftarrow \mathcal{E}(C^{*}, \rho_{E})$
\end{lstlisting}
\end{minipage}
    \caption{Distinguishing game $\text{IND}^{\text{ot-b}}_{{\sf qHE}, \mathcal{E}}$}
    \label{qHE-OT:game}
\end{figure}
A {\sf qHE} provides $\sigma(\lambda)$-IND-OT security if for all $\mathcal{E}$
adversaries, we have 
$Adv^{ind\text{-} \text{ot}}_{{\sf qHE},\mathcal{E}} (\lambda) \le \sigma(\lambda)$.
\end{definition}
\noindent\textbf{Notation for Adversaries:} Let $\mathcal{E}^{\sf qHE}$,  $\mathcal{E}^{\sf qK}$ and, $\mathcal{E}^{\sf D}$ denote the qHE adversary, the qKEM adversary, and the DEM adversary, respectively. The qKEM adversary $\mathcal{E}^{\sf qK}$ has unlimited quantum resources and is detailed in Section \ref{adversary}. However, both $\mathcal{E}^{\sf qHE}$ and $\mathcal{E}^{\sf D}$ have  polynomially bounded quantum computational power (QPT).

\begin{theorem}
\label{Thm:qHE compos}
Let {\sf qK} denote a qKEM which is $\frac{\epsilon(\lambda)}{2}$-IND-OT secure against $\mathcal{E}^{\sf qK}$ and a DEM {\sf D} is $\sigma(\lambda)$-IND-OT  $\mathcal{E}^{\sf D}$. Then, the quantum-enabled hybrid encryption {\sf qHE} is $(\epsilon(\lambda)+\sigma(\lambda))$-(IND-OT) secure against $\mathcal{E}^{\sf qHE}$.

\end{theorem}

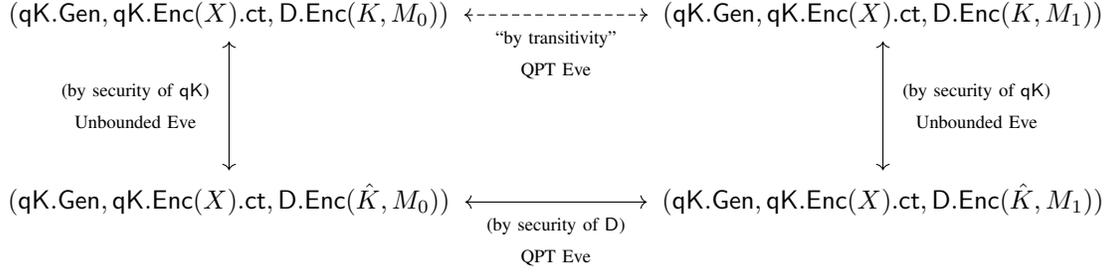
\begin{figure*}[ht]
    \centering
    \begin{tikzcd}[row sep=1.8cm, column sep=2.5cm, ampersand replacement=\&]
    ({\sf qK.Gen}, {\sf qK.Enc}(X).{\sf ct}, {\sf D.Enc}(K, M_0)) 
    \arrow[r, dashed, "{\begin{array}{c} \text{``by transitivity''} \\ \text{QPT Eve} \end{array}}"', <->, shorten <=1pt, shorten >=1pt]
    \arrow[d, "{\begin{array}{c} \text{(by security of } {\sf qK} \text{)} \\ \text{Unbounded Eve} \end{array}}"', <->, shorten <=1pt, shorten >=1pt]
        \& ({\sf qK.Gen}, {\sf qK.Enc}(X).{\sf ct}, {\sf D.Enc}(K, M_1)) 
    \arrow[d, "{\begin{array}{c} \text{(by security of } {\sf qK} \text{)} \\ \text{Unbounded Eve} \end{array}}", <->, shorten <=1pt, shorten >=1pt] \\
    ({\sf qK.Gen}, {\sf qK.Enc}(X).{\sf ct}, {\sf D.Enc}(\hat{K}, M_0)) 
    \arrow[r, "{\begin{array}{c} \text{(by security of } {\sf D} \text{)} \\ \text{QPT Eve} \end{array}}"', <->, shorten <=1pt, shorten >=1pt]
    \& ({\sf qK.Gen}, {\sf qK.Enc}(X).{\sf ct}, {\sf D.Enc}(\hat{K}, M_1)) 
    \end{tikzcd}
    \caption{High-level structure of the proof of Theorem \ref{Thm:qHE compos}.}
    \label{fig:qHE compose}
\end{figure*}

\textbf{ Proof outline}: We adapt the security proof of \cite[Chapter 10]{katz2007introduction} for {\em public key} hybrid encryption to the setting of  {\em quantum } hybrid encryption (qHE) in pre-processing model. \\
{\em Notations:} For two random variables $X$ and $Y$, we write $X \overset{\text{s}}{\equiv} Y$ to indicate that $X$ and $Y$ are statistically indistinguishable, where the distinguisher is any quantum algorithm with {\em unlimited (quantum) resources}. We also use $X \overset{\text{qc}}{\equiv} Y$ to indicate that no distinguishing algorithm  with {\em bounded quantum computation }(i.e. implementable by a quantum circuit with size  bounded by a polynomial function of the security parameter $\lambda$) can distinguish between $X$ and $Y$.

With above notations, to prove that {\sf qHE} is IND-OT secure against a QPT 
adversary  $\mathcal{E}^{\sf qHE}$ we prove that 
for any two messages $M_0$ and $M_1$ that are chosen by $\mathcal{E}^{\sf qHE}$,
we have,
\begin{align}
\label{m4}
    ([{\sf qK.Gen}, {\sf qK.Enc}(X).{\sf ct}], {\sf D.Enc}(K, M_0)) \notag \\
    \overset{\text{qc}}{\equiv} ([{\sf qK.Gen}, {\sf qK.Enc}(X).{\sf ct}], {\sf D.Enc}(K, M_1))
\end{align}
where the LHS is the adversary's view in the game \ref{qHE-OT:game} when $M_0$
is encrypted and the RHS is the corresponding view when $M_1$ is encrypted. (The bracketed part on each side is the adversary's  view when $\sf qK$ is used to generate the key $K$.)

Figure \ref{fig:qHE compose} gives a high-level outline of the proof that proceeds in three steps. First, we prove that,
\begin{align}
\label{m1}
    ({\sf qK.Gen}, {\sf qK.Enc}(X).{\sf ct}, {\sf D.Enc}(K, M_0)) \notag\\
    \overset{\text{qc}}{\equiv} ({\sf qK.Gen}, {\sf qK.Enc}(X).{\sf ct}, {\sf D.Enc}(\hat{K}, M_0))
\end{align}
where $\hat{K}$ is chosen uniformly from $\{0,1\}^{qK.Len(\lambda)}$. We will prove that using security of qKEM given by the game in Figure \ref{qKEM game}, where ${\mathcal E}^{\sf qK}$ has to distinguish between ($\rho_E, {\sf qK.Enc}(X).{\sf ct}, K$) and ($\rho_E, {\sf qK.Enc}(X).{\sf ct}, \hat{K}$), that they will be given according to a choice bit b. Next, we prove,
\begin{align}
    ({\sf qK.Gen}, {\sf qK.Enc}(X).{\sf ct}, {\sf D.Enc}(\hat{K}, M_0)) \notag \\
    \overset{\text{qc}}{\equiv} ({\sf qK.Gen}, {\sf qK.Enc}(X).{\sf ct}, {\sf D.Enc}(\hat{K}, M_1))
\end{align}
by reduction to  IND-OT security of {\sf D}, noting that $qK.ken.Len(\lambda) = D.Len(\lambda)$. Similar to Equation \ref{m1}, using security of {\sf qK}, we prove,
\begin{align}
\label{m3}
    ({\sf qK.Gen}, {\sf qK.Enc}(X).{\sf ct}, {\sf D.Enc}(K, M_1)) \notag\\
    \overset{\text{qc}}{\equiv} ({\sf qK.Gen}, {\sf qK.Enc}(X).{\sf ct}, {\sf D.Enc}(\hat{K}, M_1))
\end{align}
Finally, Equation \ref{m4} is obtained by transitivity and noting that polynomially bounded quantum algorithms  are a subset of quantum algorithms with unlimited resources and so, 
if $ X \overset{\text{s}}{\equiv} Y $ holds, then $ X \overset{\text{qc}}{\equiv} Y
$ holds also. More detailed proof follows.

\begin{proof}
We prove security using a sequence of two games, $G_0$ and $G_1$ similar to \cite{cramer2003design}. Fix an adversary $\mathcal{E}^{\sf qHE}$ and consider the original game $G_0(\lambda, b) = \text{IND}^{\text{OT-b}}_{{\sf qHE}, \mathcal{E}^{\sf qHE}}(\lambda)$ (Definition \ref{Def:qHE OT}). Our goal is to show that the advantage of $\mathcal{E}^{\sf qHE}$ in this game is upper bounded by a negligible function, $(\epsilon(\lambda) + \sigma(\lambda))$.  

In $G_0(\lambda, b)$, for a bit $b \in \{0, 1\}$, the adversary $\mathcal{E}^{\sf qHE}$ is provided with a challenge ciphertext $C^* = (C_1^*, C_2^*)$, where  
\[
(C_1^*, K) \leftarrow {\sf qK.Enc}(X) \quad \text{and} \quad C_2^* \leftarrow {\sf D.Enc}(K, M_b).
\]  

Now, we define a modified game $G_1(\lambda, b)$, which is identical to $G_0(\lambda, b)$ except that during encryption ${\sf qHE.Enc}$, a uniformly random key $\hat{K} \in \{0,1\}^{qK.Len(\lambda)}$  is used instead of $K$ to encrypt $M_b$. Specifically, the challenge ciphertext  $C^* = (C_1^*, C_2^*)$ is computed as, 
\[
(C_1^*, K) \leftarrow {\sf qK.Enc}(X) \quad \text{and} \quad C_2^* \leftarrow {\sf D.Enc}(\hat{K}, M_b).
\]  
At the end of game $G_i(\lambda, b)$ ($i = 0, 1$), the adversary $\mathcal{E}^{\sf qHE}$ outputs a bit $b_i'$. Let $T_i$ denote the event representing the success probability of $\mathcal{E}^{\sf qHE}$ in game $G_i(\lambda, b)$, defined as $\Pr[T_i] = \Pr[b_i' = b]$.

Consider the adversary $\mathcal{E}^{{\sf qK}}$ against the IND-OT security of qKEM {\sf qK}.

\noindent \textbf{Adversary}$\mathcal{E}^{{\sf qK}}$ :
     \begin{itemize}
         \item The challenger of the game $\text{IKIND}^{\text{ot-}{b}}_{{\sf qK},\mathcal{E}^{{\sf qK}}}(\lambda)$ runs the ${\sf qK.Gen}$ algorithm in presence of the adversary $\mathcal{E}^{{\sf qK}}$. $\mathcal{E}^{{\sf qK}}$ obtains a quantum state  during this generation. 
       \item $\mathcal{E}^{{\sf qK}}$ receives a challenge ciphertext key pair $(C_1^*, K^*)$  from the challenger where $K^*$ is $K$ or $\hat{K}$ according to $b=0$ or $b=1$ respectively.
        \item $\mathcal{E}^{{\sf qK}}$ provides $\rho_E$ to the adversary $\mathcal{E}^{{\sf qHE}}$. $\mathcal{E}^{\sf qHE}$ selects two message $M_0, M_1$ and provide them to $\mathcal{E}^{{\sf qK}}$.
        \item $\mathcal{E}^{{\sf qK}}$ computes $C_2^* \leftarrow {\sf D.Enc}(M_0, K^*)$ and provides $C^* = (C_1^*, C_2^*)$ to $\mathcal{E}^{{\sf qHE}}$ as the challenge ciphertext. $\mathcal{E}^{{\sf qHE}}$ outputs its guess bit $b' \in \{0, 1\}$.
        \item  $\mathcal{E}^{{\sf qK}}$ outputs whatever $\mathcal{E}^{{\sf qHE}}$ outputs.
     \end{itemize}
Note that if $b = 0$, $C^* = (C_1^*, C_2^* )$ where $C_2^* \leftarrow \text{\sf D.Enc}(M_0, K)$ and $ K = {\sf qK.Enc}(X).{\sf key}$. Hence,
{\small{\begin{equation}
\begin{split}
\label{E1}
\Pr[\mathcal{E}^{{\sf qHE}}(C_1^*,{\sf D.Enc}(M_0, K), \rho_E) =0] = \Pr[\mathcal{E}^{{\sf qK}}~\text{outputs}~0 | b = 0] 
\end{split}
\end{equation}}}
For, $b = 1$, $C^* = (C_1^*, C_2^* )$ where $C_2^* \leftarrow \text{\sf D.Enc}(M_0, \hat{K})$ where $\hat{K}$ is a uniformly random key. Hence,
{\small{\begin{equation}
\begin{split}
\label{E2}
\Pr[\mathcal{E}^{{\sf qHE}}(C_1^*,{\sf D.Enc}(M_0, \hat{K}), \rho_E) =1] = \Pr[\mathcal{E}^{{\sf qK}}~\text{outputs}~1 | b = 1] 
\end{split}
\end{equation}}}
Similarly if $\mathcal{E}^{{\sf qK}}$ computes $C_2^*$ as ${\sf D.Enc}(M_1, K^*)$,
{\small{\begin{equation}
\begin{split}
\label{E3}
\Pr[\mathcal{E}^{{\sf qHE}}(C_1^*,{\sf D.Enc}(M_1, K), \rho_E) =0] = \Pr[\mathcal{E}^{{\sf qK}}~\text{outputs}~0 | b = 0] 
\end{split}
\end{equation}}}
and 
{\small{\begin{equation}
\begin{split}
\label{E4}
\Pr[\mathcal{E}^{{\sf qHE}}(C_1^*,{\sf D.Enc}(M_1, \hat{K}), \rho_E) =1] = \Pr[\mathcal{E}^{{\sf qK}}~\text{outputs}~1 | b = 1] 
\end{split}
\end{equation}}}

Therefore,

{\small{\begin{align*}
\begin{split}
& 
\Pr[T_0] - \Pr[T_1] \\
& = \left\{\frac{1}{2} \Pr[\mathcal{E}^{{\sf qHE}}(C_1^*,{\sf D.Enc}(M_0, K), \rho_E) =0] \right. \nonumber\\
& ~~~~~~~~~~~~~~~~~~~~\left. \quad + \frac{1}{2} \Pr[\mathcal{E}^{{\sf qHE}}(C_1^*,{\sf D.Enc}(M_1, K), \rho_E) =1]\right\}\\
& ~~~~~ - \left\{\frac{1}{2} \Pr[\mathcal{E}^{{\sf qHE}}(C_1^*,{\sf D.Enc}(M_0, \hat{K}), \rho_E) =0] \right. \nonumber\\
& ~~~~~~~~~~~~~~~~~~~~ \left. \quad + \frac{1}{2} \Pr[\mathcal{E}^{{\sf qHE}}(C_1^*,{\sf D.Enc}(M_1, \hat{K}), \rho_E) =1]\right\} \\
& = \left\{\frac{1}{2} \Pr[\mathcal{E}^{{\sf qHE}}(C_1^*,{\sf D.Enc}(M_0, K), \rho_E) =0] \right. \nonumber\\
& ~~~~~~~~~~~~~~ \left. \quad + \frac{1}{2} \Pr[\mathcal{E}^{{\sf qHE}}(C_1^*,{\sf D.Enc}(M_0, \hat{K}), \rho_E) =1] - \frac{1}{2}\right\}\\ 
& ~~+ \left\{\frac{1}{2} \Pr[\mathcal{E}^{{\sf qHE}}(C_1^*,{\sf D.Enc}(M_1, K), \rho_E) =0] \right. \nonumber \\
& ~~~~~~~~~~~~~~~\left. \quad + \frac{1}{2} \Pr[\mathcal{E}^{{\sf qHE}}(C_1^*,{\sf D.Enc}(M_1, \hat{K}), \rho_E) =1] - \frac{1}{2}\right\}\\
& = \left\{\frac{1}{2}\Pr[\mathcal{E}^{{\sf qK}}~\text{outputs}~0 | b = 0] + \frac{1}{2}\Pr[\mathcal{E}^{{\sf qK}}~\text{outputs}~1 | b = 1]-\frac{1}{2}\right\}\\
& ~+\left\{\frac{1}{2}\Pr[\mathcal{E}^{{\sf qK}}~\text{outputs}~0 | b = 0] + \frac{1}{2}\Pr[\mathcal{E}^{{\sf qK}}~\text{outputs}~1 | b = 1]-\frac{1}{2}\right\}  \\
&  \le \frac{1}{2}Adv^{ikind\text{-}ot}_{{\sf qK},\mathcal{E}^{{\sf qK}}}(\lambda) + \frac{1}{2}Adv^{ikind\text{-}ot}_{{\sf qK},\mathcal{E}^{{\sf qK}}}(\lambda) = Adv^{ikind\text{-}ot}_{{\sf qK},\mathcal{E}^{{\sf qK}}}(\lambda)
\end{split}
\end{align*}}}

\noindent The second-to-last equality follows from equations \ref{E1}, \ref{E2}, \ref{E3}, and \ref{E4}. The final inequality follows from the fact that the unbounded quantum adversary $\mathcal{E}^{{\sf qK}}$ is at least as powerful as the quantum polynomial-time (QPT) adversary $\mathcal{E}^{{\sf qHE}}$. Hence,
\begin{equation}
\label{Eq:KEM adv}
    \Pr[T_0] \le \Pr[T_1] + Adv^{ikind\text{-}ot}_{{\sf qK},\mathcal{E}^{{\sf qK}}}(\lambda) 
\end{equation}




 

\noindent Now consider the adversary $ \mathcal{E}^{\sf D}$ against IND-OT security of DEM ${\sf D}$. \\
  
  \noindent \textbf{Adversary} $\mathcal{E}^{{\sf D}}$:
  \begin{itemize}
      \item $\mathcal{E}^{{\sf D}}$ is involved in the game $\text{IND}^{\text{ot-}{b}}_{{\sf D},\mathcal{E}^{{\sf D}}}(\lambda)$ where the game challenger has generated a random  symmetric key $\hat{K}\stackrel{\$} \gets {\sf D.Gen}(1^\lambda)$.
     \item $\mathcal{E}^{{\sf D}}$ runs {\sf qHE.Gen} algorithm in the presence of the adversary $\mathcal{E}^{\sf qHE}$. $\mathcal{E}^{{\sf D}}$ keeps $X$, $Y$ and $\mathcal{E}^{\sf qHE}$ gets a quantum state $\rho_E$.
      \item $\mathcal{E}^{\sf qHE}$ outputs two messages $M_0$, $M_1$ to $\mathcal{E}^{{\sf D}}$.
      \item $\mathcal{E}^{{\sf D}}$ sends $M_0$, $M_1$ to the challenger of the game $\text{IND}^{\text{ot-}{b}}_{{\sf D},\mathcal{E}^{{\sf D}}}(\lambda)$.
      \item The challenger constructs a ciphertext $C_2^* \stackrel{\$} \gets (\hat{K}, M_b)$ and sends $C_2^*$ to $\mathcal{E}^{{\sf D}}$.
      \item $\mathcal{E}^{{\sf D}}$ generates $(C_1^*, K) \stackrel{\$} \gets {\sf qK.Enc}(X)$ and sends $C^* = (C_1^*, C_2^*)$ to $\mathcal{E}^{\sf qHE}$.
      
      \item $\mathcal{E}^{\sf qHE}$ outputs a guess bit $b' \in \{0,1\}$.
      \item $\mathcal{E}^{{\sf D}}$ outputs the same bit. 
      
  \end{itemize}
Note that, when $b=0$, 
{\small{\begin{equation}
\begin{split}
\Pr[\mathcal{E}^{{\sf qHE}}(C_1^*,{\sf D.Enc}(M_0, \hat{K}), \rho_E) =0] =
\Pr[\mathcal{E}^{{\sf D}}~\text{outputs}~0 | b = 0] 
\end{split}
\end{equation}}}


For $b =1$,
{\small{\begin{equation}
\begin{split}
\Pr[\mathcal{E}^{{\sf qHE}}(C_1^*,{\sf D.Enc}(M_1, \hat{K}), \rho_E) =1] = \Pr[\mathcal{E}^{{\sf D}}~\text{outputs}~1 | b = 1] 
\end{split}
\end{equation}}}
Hence,
{\small{\begin{equation}
\label{Eq: DEM adv}
\begin{split}
&
\Pr[T_1]=\frac{1}{2}\Pr[\mathcal{E}^{{\sf qHE}}(C_1^*,{\sf D.Enc}(M_0, \hat{K}), \rho_E) =0] + \nonumber\\
& \quad ~~~~~~~~~~\frac{1}{2} \Pr[\mathcal{E}^{{\sf qHE}}(C_1^*,{\sf D.Enc}(M_1, \hat{K}), \rho_E) =1]\\
& = \frac{1}{2}\Pr[\mathcal{E}^{{\sf D}}~\text{outputs}~0 | b = 0] + \frac{1}{2}\Pr[\mathcal{E}^{{\sf D}}~\text{outputs}~1 | b = 1]\\
& = \frac{1}{2} + \frac{1}{2} Adv^{ind\text{-} \text{ot}}_{{\sf D},\mathcal{E}^{{\sf D}}}(\lambda)
\end{split}
\end{equation}}}




From Equations \ref{Eq:KEM adv} and \ref{Eq: DEM adv} we have,
\begin{equation}
    \Pr[T_0] \le \frac{1}{2} + \frac{1}{2} Adv^{ind\text{-} \text{ot}}_{{\sf D},\mathcal{E}^{{\sf D}}}(\lambda) + Adv^{ikind\text{-}ot}_{{\sf qK},\mathcal{E}^{{\sf qK}}}(\lambda)
\end{equation}
Therefore,
\begin{align*}
    Adv^{ind\text{-} \text{ot}}_{{\sf qHE}, \mathcal{E}^{\sf qHE}} (\lambda) 
    & = 2 \left|\Pr[T_0] - \frac{1}{2} \right|\\[10pt]
    & \le 2 \times Adv^{ikind\text{-}ot}_{{\sf qK},\mathcal{E}^{{\sf qK}}}(\lambda)+ Adv^{ind\text{-} \text{ot}}_{{\sf D},\mathcal{E}^{{\sf D}}}(\lambda) \\[10pt]
    & \le 2 \times \frac{\epsilon(\lambda)}{2} + \sigma(\lambda) = \epsilon(\lambda) + \sigma(\lambda) 
\end{align*}


This concludes the proof.
\end{proof}

\section{Construction of qKEM}
\label{instantiation:qKEM}
We construct a q-KEM {\sf CqK} using the 
QKD construction in \cite[Figure 11]{portmann2022security} when the quatum adversary uses 
{\em collective strategy}. 

\subsection{qKEM Algorithms}
\label{q-KEM Alg}
The qKEM {\sf CqK}, with security parameter $\lambda$ and based on the QKD protocol in Section \ref{QKD}, consists of the algorithms {\sf CqK = (CqK.Gen$_\mathcal{E}$, CqK.Enc, CqK.Dec)}. The key length is $\{0,1\}^\ell$. $\mathbb{H} = \{h_S |~h_S: \{0,1\}^n \rightarrow \{0,1\}^{\ell}, 0 < \ell \le n \}_{S \in \mathcal{S}}$ and $\mathbb{H}_1 = \{h_S |~h_S: \{0,1\}^n \rightarrow \{0,1\}^{t}, 0 < t \le n \}_{S \in \mathcal{S}^\prime}$ be two families of 2-UHF parameterized by a seed $S$ from $\mathcal{S}$ and $\mathcal{S}^\prime$ respectively.
\begin{enumerate}

\item  $({\bf X}_A, {\bf X}_B ,\rho_E) \longleftarrow $ {\sf CqK.Gen}$_\mathcal{E}(1^\lambda)$: Alice and Bob proceed with the following steps:
\begin{enumerate}
    \item Execute the State Preparation, transmission and, Sifting algorithm (see Algorithm \ref{State prep and Sifting}). At the end of this step, Alice and Bob have strings ${\bf X}_A^N$ and ${\bf X}_B^N$ respectively.
    \item Execute the Parameter Estimation protocol (Algorithm \ref{p.estimation}). Following this step, Alice and Bob possess raw key strings ${\bf X}_A$ and ${\bf X}_B$ respectively.
\end{enumerate}
Steps a and b are in presence of Eve who interacts with the quantum channel and observes  the public authenticated channel. Eve's quantum state is $\rho_E$.

\item $({\bf K}_A,C) \longleftarrow $ {\sf CqK.Enc}$({\bf X}_A)$: Alice does the following steps:
\begin{enumerate}
    \item Chooses $S^\prime \in_R \mathcal{S}^\prime$ and $S \in_R \mathcal{S}$.
    \item Derives ${\bf K}_A = h_S({\bf X}_A)$.
    \item Computes $C$ as follows:
    \begin{itemize}
        \item Calculates ${\bf V} = h_{S^\prime}({\bf X}_A)$ and $ {\bf W} = enc({\bf X}_A)$.
        \item Sets $C = ({\bf W}, S, S^\prime,{\bf V})$.
    \end{itemize}
\end{enumerate}

\item $ ({\bf K}_B)\longleftarrow$ {\sf CqK.Dec}$({\bf X}_B, C)$: Bob parses C as $({\bf W}, S, S^\prime, {\bf V})$.
    \begin{enumerate}
       \item Runs $dec({\bf X}_B, {\bf W})$ as 
       in Algorithm \ref{I.reconci} to obtain ${\bf X}_B^\prime$.
       \item If ${\bf V} = h_{S^\prime}({\bf X}_B^\prime)$, then output ${\bf K}_B = h_S({\bf X}_B^\prime)$ ; else output $\perp$.
       
       \end{enumerate}
\end{enumerate}


\noindent \textbf{Correctness}: The correctness of {\sf CqK} follows from the correctness of QKD (see \ref{QKD}).

 




\noindent \textbf{Security}: We consider the security of the above construction against a computationally unbounded quantum adversary $\mathcal{E}^{\sf cqK}$ without any oracle access (OT attacker).
\begin{theorem}
\label{secproof:cqaqkem}
The qKEM {\sf CqK }in Section \ref{q-KEM Alg} is $\epsilon(\lambda)$-(ND-OT secure according to definition \ref{def:qEMOT} if and only if the QKD in section \ref{QKD} is $\epsilon(\lambda)$-secret according to definition \ref{Def:QKD Secret}.

\end{theorem}

\begin{proof}

($\Rightarrow$) We first show that if {\sf CqK} is secure, then the QKD must be secret. Using the security definition of qKEM we have,
\begin{align}
\scriptsize
Adv^{ikind\text{-} \text{ot}}_{{\sf CqK},\mathcal{E}} (\lambda) 
&= \big| \Pr[\text{IKIND}^{\text{ot-0}}_{{\sf CqK},\mathcal{E}}(\lambda) = 1] \notag \\
&\quad - \Pr[\text{IKIND}^{\text{ot-1}}_{{\sf CqK},\mathcal{E}}(\lambda) = 1] \big| \notag \\
&\leq \epsilon(\lambda)
\end{align}
For $b = 0$, the adversary $\mathcal{E}^{\sf cqK}$ has the quantum state $\rho_E$ and receives $({\bf K}^*, C^*)$, where ${\bf K}^* ={\sf CqK.Enc}({\bf X}_A).key$ and $C^* = {\sf CqK.Enc }({\bf X}_A).ct$.
That is for $b=0$,  $\mathcal{E}^{\sf cqk}$ has the quantum state,
$\rho_{E {C^*} [{\sf CqK.Enc}({\bf X}_A).key]}$, and for $b=1$, it has the state,
$\rho_{E{C^*}\hat{K}}$,
where  $\hat{K}$  is a uniform random string of length $\ell$
(which is the length of $ {\sf CqK.Enc}({\bf X}_A).key$). 
Thus 
\begin{equation}\label{adv-SD}
    Adv^{ikind\text{-} \text{ot}}_{{\sf CqK},\mathcal{E}} (\lambda)\leq 
    \frac{1}{2}\norm{\rho_{E C^{*} {\bf K}^{*}}- \rho_{E C^{*}} \otimes  \omega_{{\bf K}^{*}} }_1
    \end{equation}
This is because 
$SD(f({\bf X}), f({\bf Y}) \leq  SD ({\bf X}, {\bf Y})$, and so any adversary (algorithm) $\mathcal{E}^{\sf CqK}$ has distinguishing probability that is bounded by  $(\ref{adv-SD})$.
Here $\omega_{{\bf K}^{*}} = \frac{1}{2^{\ell}}\sum _{{\bf k} \in \{0,1\}^{\ell}} \ket{{\bf k}}\bra{{\bf k}}$ is the maximally mixed state over the key space $\{0,1\}^\ell$.

If the QKD is not $\epsilon(\lambda)$-secret, then there is an adversary $\mathcal{E}^{\sf qkd}$ that, according to definition \ref{Def:QKD Secret}, receives  the state of $\mathcal{E}^{\sf cqK}$ 
and  has,
\begin{equation} 
\label{qkd-sec}
 \frac{1}{2}\norm{\rho_{E C^{*} {\bf K}^{*}}- \rho_{E C^{*}} \otimes  \omega_{{\bf K}^{*}} }_1 \geq \epsilon(\lambda)
\end{equation}
Using \cite[Lemma 7]{portmann2022security}, we have:
  {\small{  \begin{equation}
          \frac{1}{2}\norm{\rho_{E C^{*} {\bf K}^{*}}- \rho_{E C^{*}} \otimes  \omega_{{\bf K}^{*}} }_1 = \max_{\{E_z\}} SD((\bold{Z}, C^{*}, {\bf K}^{*}), (\bold{Z}, C^{*}, U_\ell))
    \label{eq relation POVM}
    \end{equation}}}
where $\{E_z\}$ is a POVM that is applied on the quantum system E, $\bold{Z}$ is the random variable corresponding to the POVM output distribution, and $U_\ell$ is the uniformly distributed random variable over $\{0,1\}^\ell$. That is there is a POVM for which 
\begin{equation}
    SD((\bold{Z}, C^{*}, {\bf K}^{*}), (\bold{Z}, C^{*}, U_\ell)) \geq \epsilon(\lambda)
\end{equation}
Using properties of statistical distance (see \cite{oded2001foundations,sharifian2021information}) there is a subset $\mathcal{W} \subset \mathcal{Z} \times \mathcal{C} \times \mathcal{K}$ for which
\begin{align}
\label{eq:sd_bound}
SD((\bold{Z}, C^{*}, {\bf K}^{*}), (\bold{Z}, C^{*}, U_\ell)) 
&= \max_{\cal W} \big| \Pr((\bold{Z}, {\bf C}^{*}, {\bf K}^{*}) \in \mathcal{W}) \notag \\
&\quad - \Pr((\bold{Z}, {\bf C}^{*}, {\bf \hat{K}}) \in \mathcal{W}) \big| \notag \\
&\geq \epsilon(\lambda)
\end{align}

This defines a distinguishing algorithm for ${\cal E}^{\sf CqK}$ (in game \ref{qKEM game} the algorithm outputs  $b'=1$ if $( \bold{Z}, {\bf C}^{*}, {\bf K}^{*}) \in {\cal W}$, and $0$, otherwise) with advantage greater than $\epsilon(\lambda)$ which is a contradiction.

\noindent ($\Leftarrow$) Next we show that if the QKD is $\epsilon(\lambda)$-secret, then {\sf CqK} must be $\epsilon(\lambda)$-secure.

Using secrecy definition of QKD,  we have 
\begin{equation*}
    \frac{1}{2} \norm{\rho_{E C {\bf K} }-\rho_{
    E C} \otimes {\omega_{\bf K} }
  }_1 \leq \epsilon (\lambda).
\end{equation*}
This establishes a bound on 
     $\frac{1}{2}\norm{\rho_{E C^{*} {\bf K}^{*}}- \rho_{E C^{*}} \otimes  \omega_{{\bf K}^{*}} }_1 \leq \epsilon(\lambda)$
representing the quantum states of ${\cal E}^{\sf cqK}$ for $b = 0$ and $b = 1$, which bounds the advantage of ${\cal E}^{\sf cqk}$ using expression (\ref{adv-SD}). That is the qKEM {\sf CqK}
will be secure.

\end{proof}

\section{Concluding Remarks}
\label{conclusion}
We extend the well-established framework of HPKE in public key setting, and its extension HE in correlated randomness setting, to qHE that allows the key encapsulation mechanism to be based on QKD. Hybrid encryption systems aim to securely compose two cryptographic primitives, a KEM, which establishes a shared secret key, and a DEM, that uses the established key in a symmetric key-based encryption, to encrypt the data. In each setting, KEM and DEM have well-defined game-based security definitions, and the composition theorem proves the security of the hybrid encryption using a game-based security definition.

Compared to the two previous settings, the challenge of defining and proving the security of qKEM and qHE is that the setup phase of qKEM (and so qHE) is in the presence of an {\em active} adversary. qKEMs provide information-theoretic security for the established key and, compared to iKEM, have the important advantage of not requiring any assumption on the adversary's information. That allows the adversary's information to be correctly estimated during the setup phase.

Our work was motivated by providing explicit security proof for using a QKD established key with a symmetric key based encryption system and so we focus on basic indistinguishability-based security games in which the adversary does not have query access to KEM and DEM oracles (as appropriate) as considered in HPKE and HE in correlated randomness model.

In practice, a QKD protocol can generate more key bits than what is required by a one-time symmetric key encryption system. Modeling and proving the security of using the established key for multi-time hybrid encryption is an interesting direction for future work.

\bibliographystyle{plain}
\bibliography{ref}

\end{document}